\documentclass[10pt,aps,prx,twocolumn,footinbib,superscriptaddress,floatfix]{revtex4-2}
\usepackage{graphicx}
\usepackage{indentfirst}
\usepackage{siunitx}
\usepackage{physics}
\usepackage{braket}
\usepackage{float}
\usepackage{mathtools}
\usepackage{epstopdf}
\usepackage{footnote}
\usepackage{esint}
\usepackage{comment}
\usepackage{color}
\usepackage[T1]{fontenc}
\usepackage{amsfonts}
\usepackage{footmisc}
\usepackage{scrextend}
\usepackage{multirow}
\usepackage[hyperfootnotes=false]{hyperref}
\usepackage[acronym]{glossaries}
\usepackage[english]{babel}
\usepackage{url}
\usepackage{bm}
\usepackage{tikz}
\usepackage{etoolbox}
\usepackage{enumerate}
\usepackage{soul}
\usepackage{amssymb}
\definecolor{darkblue}{rgb}{0,0,0.5}
\hypersetup{
    colorlinks=true,
    linkcolor=black,
    filecolor=blue,
    citecolor=darkblue,  
    urlcolor=black,
}

\usepackage[normalem]{ulem}

\apptocmd{\sloppy}{\hbadness 9999\relax}{}{}

\DeclareMathOperator{\E}{\mathbb{E}}

\DeclarePairedDelimiterX\parvertent[2]{\lparen}{\rparen}%
{#1\,\delimsize\vert\,\mathopen{}#2}

\newtheorem{theorem}{Theorem}

\newtheorem{corollary}[theorem]{Corollary}

\newtheorem{lemma}[theorem]{Lemma}

\newenvironment{proof}[1][Proof]{\noindent\textbf{#1.} }{\ \rule{0.5em}{0.5em}}

\newcommand{\calA}{{\cal A}}

\newcommand{\calE}{{\cal E}}
\newcommand{\calF}{{\cal F}}

\newcommand{\calL}{{\cal L}}

\newcommand{\calO}{{\cal O}} 
\newcommand{\calP}{{\cal P}}

\newcommand{\calH}{{\cal H}}

\newcommand{\calU}{{\cal U}}

\newcommand{\1}{^{(1)}}

\newcommand{\state}[1]{\ketbra{#1}{#1}}

\newcommand{\bI}{\boldsymbol I}

\newcommand{\bfz}{\mathbf{z}}

\def\be{\begin{equation}}
\def\ee{\end{equation}}
\def\ba{\begin{eqnarray}}
\def\ea{\end{eqnarray}}

\newcommand{\QZ}[1]{{{\textcolor{black}{#1}}}}

\usepackage[normalem]{ulem}
\usepackage{xr}
\begin{document}

\title{\QZ{Holographic deep thermalization \\
for secure and efficient quantum random state generation}}

\author{Bingzhi Zhang}
\thanks{These authors contributed equally.}
\email{bingzhiz@usc.edu}
\affiliation{ Department of Physics and Astronomy, University of Southern California, Los
Angeles, California 90089, USA
}
\affiliation{
Ming Hsieh Department of Electrical and Computer Engineering, University of Southern California, Los
Angeles, California 90089, USA
}

\author{Peng Xu}
\thanks{These authors contributed equally.}
\affiliation{Department of Statistics, University of Illinois at Urbana-Champaign, Champaign, IL 61820, USA}

\author{Xiaohui Chen}
\affiliation{Department of Mathematics, University of
Southern California, Los
Angeles, California 90089, USA}

\author{Quntao Zhuang}
\email{qzhuang@usc.edu}
\affiliation{
Ming Hsieh Department of Electrical and Computer Engineering, University of Southern California, Los
Angeles, California 90089, USA
}
\affiliation{ Department of Physics and Astronomy, University of Southern California, Los
Angeles, California 90089, USA
}

\begin{abstract}

Quantum randomness, especially random pure states, \QZ{underpins} fundamental questions like black hole physics and quantum complexity, as well as in practical applications such as quantum device benchmarking and quantum advantage certification. The conventional approach for generating genuine random states, \QZ{known as} `deep thermalization', faces significant challenges, including scalability issues due to the need for a large ancilla system and susceptibility to attacks, as demonstrated in this work.
We introduce holographic deep thermalization, a secure and hardware-efficient quantum random state generator.
\QZ{Via a sequence of scrambling-measure-reset processes}, it continuously trades space with time, and substantially reduces the required ancilla size to as small as a system-size-independent constant; at the same time, it guarantees security by \QZ{eliminating} quantum correlation between the data system and attackers.
Thanks to the resource reduction, our circuit-based implementation on IBM Quantum devices achieves genuine $5$-qubit random state generation utilizing only a total of $8$ qubits.
\end{abstract}

\maketitle


\section{Introduction}

Randomness is a fundamental concept in science, underpinning the backbone of fields such as statistics, information theory, dynamical systems, and thermodynamics. Its applications are extensive, spanning cryptography, Monte Carlo simulations, and machine learning. True randomness originates from quantum physics, where the act of measurement collapses quantum superpositions into inherently unpredictable outcomes.
Random quantum states, in particular, have become central to modern quantum science, playing a crucial role in understanding foundational theories, such as black hole physics~\cite{hayden2007black}, quantum chaos~\cite{roberts2017chaos}, and the complexity of quantum circuits~\cite{brandao2021models}. Random quantum states are equally essential in practical applications, including quantum device benchmarking~\cite{cross2019validating}, quantum advantage certification~\cite{arute2019quantum, tillmann2013experimental}, quantum cryptography~\cite{pirandola2020advances} and quantum learning~\cite{huang2020predicting, elben2023randomized}. 
The unique type of random quantum states necessary in these scenarios is the Haar-random ensemble \QZ{or its approximations. \QZ{The} Haar-random ensemble is} uniformly distributed over the entire Hilbert space, which represents typical quantum states with high complexity and volume-law entanglement. For the same reason, though, \QZ{the} Haar ensemble is also difficult to sample.

Several approaches have been developed to generate approximate Haar-random states, \QZ{formally known as the $K$-design~\cite{roberts2017chaos}}. \QZ{The} early approach \QZ{based on} local random unitary circuits requires a large \QZ{circuit} depth linear in system size~\cite{brandao2016local, harrow2023approximate}. 
In addition, \QZ{different samples are generated by} by randomly adjusting the quantum gates, leading to significant overhead in quantum circuit compilation and experimentation. \QZ{At the same time, the quality of the random states generated by the approach relies on the quality of classical randomness.} 

Recently, \emph{deep thermalization} (DT)~\cite{ho2022exact,ippoliti2022solvable,cotler2023emergent,ippoliti2023dynamical,choi2023preparing} suggests a promising approach to generate genuine random states from partial quantum measurements, with both initial state and quantum dynamics fixed. Random states emerge within a small subsystem via projective measurements over the rest \QZ{of the system} when the whole system sufficiently thermalizes. However, such approaches require a large number of ancilla qubits linearly growing with the system size, \QZ{posing} an experimental challenge considering near-term quantum devices~\cite{choi2023preparing}. Furthermore, we prove that DT is vulnerable to \QZ{an} entanglement attack---\QZ{the} initial quantum correlation between the data system and \QZ{the} attacker \QZ{decreases by} at most a constant \QZ{amount} regardless of \QZ{the} ancilla size.

In this work, we introduce the \emph{holographic deep thermalization} (HDT), a secure and hardware-efficient paradigm for generating random quantum states with only a \emph{constant} number of ancilla qubits, via trading space with time~\cite{anand2023holographic}. In each time step, a fixed high-complexity unitary, similar to that used in DT but acting on a much smaller system, is applied, and then the ancillae are measured and reset. 
We rigorously prove the convergence of HDT in terms of the frame potential and show that quantum correlation between the data system and any attacker is erased step-by-step to guarantee security. Thanks to the resource reduction, our circuit-based implementation on IBM Quantum devices achieves genuine $2$-qubit random state generation utilizing only a $4$-qubit system, while previous work adopts a 10-atom system~\cite{choi2023preparing}. When adopting 5 data qubits and 3 ancilla qubits, we also extend the dimension of random states to $d_A=32$, representing a six-fold increase compared to the previous experiment~\cite{choi2023preparing}.

Furthermore, via balancing the ancilla size and number of time steps, HDT achieves a continuous trade-off between the total quantum circuit size and the ancilla size, as we demonstrate \QZ{both} analytically and experimentally. Besides the minimum-space operating point, we also identify a minimum-quantum-circuit-size operating point where \QZ{the} ancilla size equals the system size.

\begin{figure}[t]
    \centering
    \includegraphics[width=0.45\textwidth]{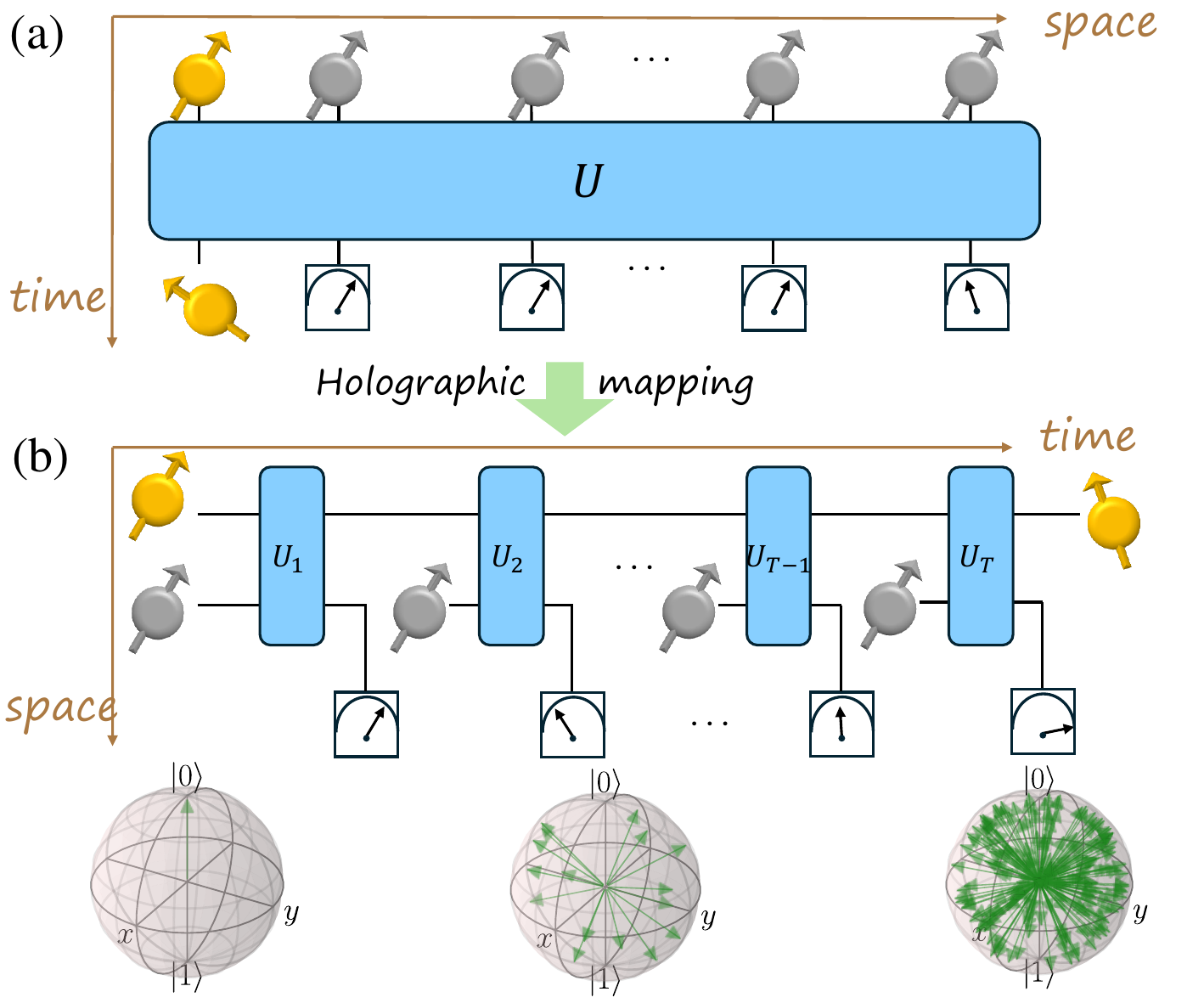}
    \caption{\QZ{\bf Conceptual plot of the schemes.} (a) Deep thermalization  with a single chaotic unitary $U$. (b) Holographic deep thermalization (HDT) with a sequence of chaotic unitaries $\{U_t\}_{t=1}^{T}$. Yellow and gray \QZ{represent} data and ancilla systems. \QZ{In the bottom \QZ{panel}, we plot the {\em projected ensemble} in the Bloch sphere, starting from the initial single-state input toward the late-time distribution to represent the growth of complexity and randomness of the generated state ensemble in HDT.}}
    \label{fig:scheme}
\end{figure}

\section{Problem set-up and protocol design}
Generation of genuinely quantum random states in DT relies on quantum measurement projection. For a fixed pure quantum state $\ket{\Psi}_{AB}$ on a joint system of data $A$ and ancilla $B$, the \emph{projected ensemble} $\calE(\Psi)$ on subsystem $A$ is generated via a projective measurement on subsystem $B$, e.g., with projectors $\{\ketbra{z}{z}_B\}$ in the computational basis~\cite{cotler2023emergent}. Conditioned on the measurement result $z$ on $B$, data system $A$ \QZ{collapses to} a pure state $\ket{\psi_z} \propto {{}_B\braket{z|\Psi}_{AB}}$, with probability $p_z=\tr(\ketbra{\Psi}{\Psi}_{AB}\ketbra{z}{z}_B)$.

To quantify \QZ{how close} a state ensemble $\calE(\Psi)$ \QZ{is} to the Haar-random ensemble up to the $K$-th moment, we consider its frame potential~\cite{ippoliti2023dynamical}, defined as
\be\label{eq:fp_def}
\begin{aligned}
\calF^{(K)}&\equiv \sum_{z, z'} p_z p_{z'}\abs{\braket{\psi_z|\psi_{z'}}}^{2K} \\
&\ge \calF_\text{Haar}^{(K)} = \binom{d_A + K-1}{d_A - 1}^{-1},
\end{aligned}
\ee
where $d_A=2^{N_A}$ is the dimension of the $N_A$-qubit data system $A$.
The frame potential quantifies the \QZ{inverse dimension of the} occupied symmetric space and is thus
lower bounded by $\calF_{\rm Haar}^{(K)}$---\QZ{the frame potential of the Haar-random ensemble}.

When an ensemble achieves the minimum frame potential $\calF_\text{Haar}^{(K)}$, one regards the ensemble as the $K$-design~\cite{roberts2017chaos, cotler2023emergent, ho2022exact, ippoliti2023dynamical}, which approximates the Haar-random ensemble to at least the $K$-th moment.
An $\epsilon$-approximate $K$-design state ensemble $\calE$ is defined via the relative deviation of frame potential as 
$
\delta^{(K)} \equiv \calF^{(K)}/\calF_\text{Haar}^{(K)} - 1 \le \epsilon
$ (see Methods).


\begin{figure*}[t]
    \centering
    \includegraphics[width=0.85\textwidth]{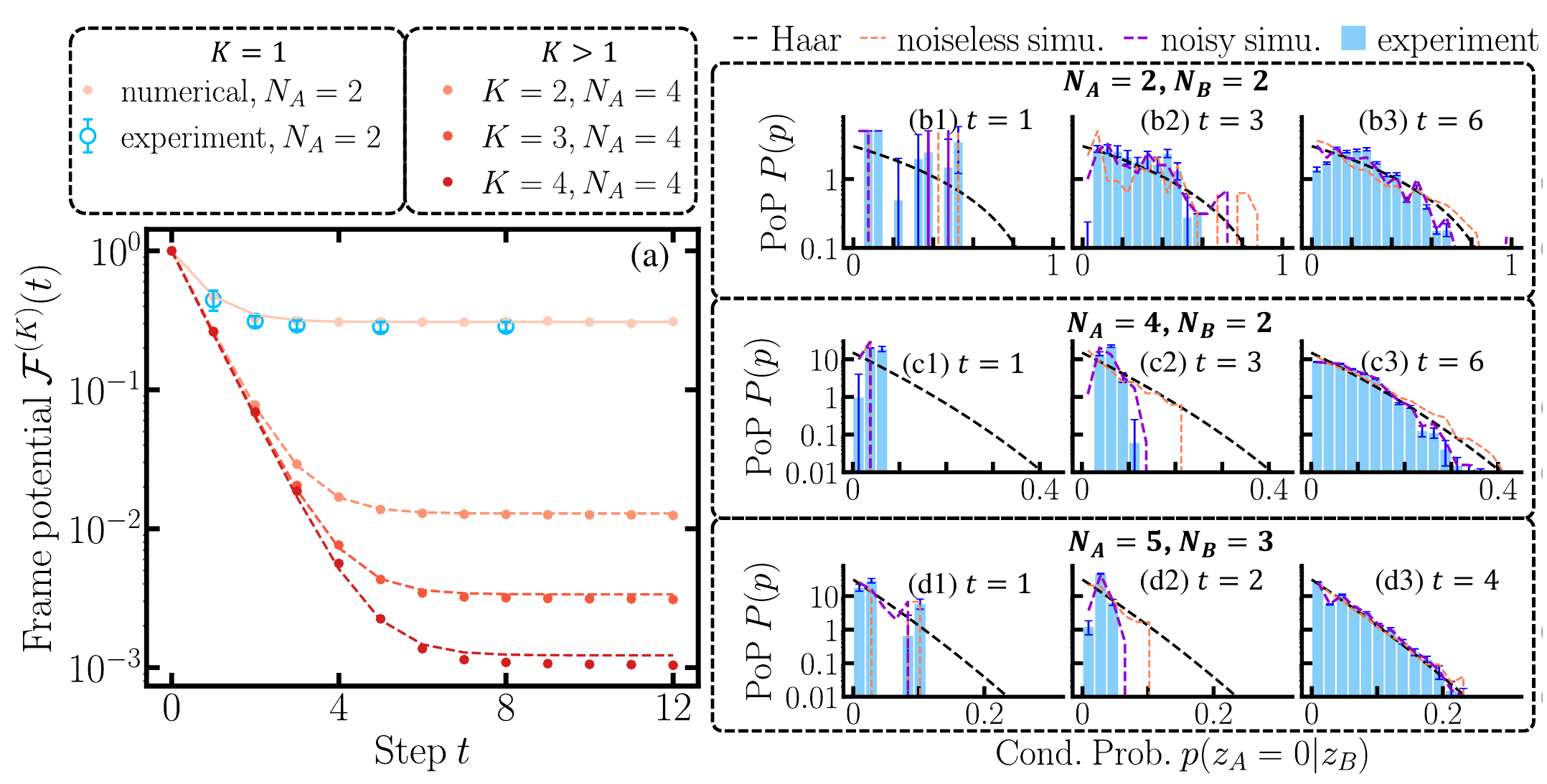}
    \caption{\QZ{\bf Dynamics of frame potential and probability distribution of probability (PoP) in HDT.} In (a), the red dots from light to dark represent numerical simulation results of $\calF^{(K)}(t)$ in HDT from $K=1$ to $K=4$. Here $N_B=2$. The red solid line represents first-order theoretical prediction of Eq.~\eqref{F1_exact} in Theorem~\ref{first_order_F} and red dashed lines from light to dark represent \QZ{the asymptotic lower bounds of the higher-order frame potential} in \QZ{Ineq.~\eqref{high_order_F}} for corresponding $K$. Blue dots with errorbars show IBM Quantum experimental results for $\calF^{(1)}(t)$.   
    In the right panel, we plot the PoP at different time steps in HDT with various choices of data and ancilla system \QZ{size} $N_A, N_B$ (listed in titles). 
    Blue bars represent the IBM Quantum experimental results of average PoP with errorbars indicating fluctuations.
    Black dashed lines are PT distribution for Haar theory, orange and purple dashed lines show the noiseless and noisy \QZ{device} simulations for reference.} 
    \label{fig:Fevo}
\end{figure*}


Previous studies on the conventional DT~~\cite{ho2022exact,ippoliti2022solvable,cotler2023emergent,ippoliti2023dynamical,choi2023preparing} (cf. Fig.~\ref{fig:scheme}a) focused on the {\em projected ensemble} obtained from a single fixed unitary on a trivial initial state. It requires $N_B = \Omega(K N_A + \log (1/\epsilon))$ ancilla qubits to achieve the $\epsilon$-approximate $K$-design\QZ{~\cite{ippoliti2022solvable, ippoliti2023dynamical}}, which creates an overhead growing with the system size\QZ{---}especially for high-order state design generation.

We propose HDT (cf. Fig.~\ref{fig:scheme}b), which extends the time direction to $T$ discrete steps to compensate for the reduction in space---in the same spirit of holographic quantum simulation~\cite{anand2023holographic}.
HDT also begins with a trivial state, e.g., $\ket{\psi_0} = \ket{0}^{\otimes N_A}$. In each of the $T$ time steps, $N_B$ ancilla qubits initialized in $\ket{0}^{\otimes N_B}$ interact with \QZ{the} data system $A$ through a fixed \QZ{high-complexity} unitary $U_t$, which can be implemented by either a randomly sampled fixed deep quantum circuit~\cite{nahum2017quantum, nahum2018operator} or chaotic Hamiltonian evolution~\cite{cotler2023emergent, ho2022exact}. Following projective measurements on the ancilla system, the \emph{projected} ensemble $\calE_t$ of the data system at step $t$ is fed towards the next time step as input.

\section{Convergence towards the Haar ensemble}
Assuming that each $U_t$ is a typical unitary, one expects the \emph{projected ensemble} $\calE_t$ to converge towards the fixed point of \QZ{the} Haar ensemble, as $t$ increases. Indeed, our HDT only requires $K + \log_2(1/\epsilon)$ constant number of ancilla qubits and a time step of $K N_A/[K+\log_2(1/\epsilon)]$ to achieve $\epsilon$-approximate $K$-design.

Below, we evaluate the frame potential to quantify the rate of convergence. Denoting the data and ancilla system \QZ{dimensions} as $d_A=2^{N_A}$ and $d_B=2^{N_B}$, we begin with the exact solution for the typical first-order frame potential \QZ{and then proceed to the higher-order frame potential. A proof sketch is provided in Methods and the complete proof can be found in Appendix \ref{app:F_dynamics}.}
\begin{theorem}
     The typical first-order frame potential $\calF^{(1)}(t)$ of the projected ensemble $\calE_t$ at step $t$ in the holographic deep thermalization \QZ{(with each $U_t$ typical for at least $2$-design)} is
     \begin{equation}\label{F1_exact}
         \begin{aligned}
         \calF^{(1)}(t)
         &=\frac{(d_A-1) (d_A d_B-1)}{d_A^2 d_B+1}\left(\frac{\left(d_A^2-1\right) d_B}{d_A^2 d_B^2-1}\right)^t \\
         &\quad + \frac{d_A^2(d_B+1)}{d_A^2 d_B+1} \calF_\textup{Haar}^{(1)}.
         \end{aligned}
     \end{equation}
    In the thermodynamic limit of \QZ{the} data system ($d_A \to \infty$), we have
    $
    \QZ{\calF^{(1)}(t) = d_B^{-t} + (1+1/d_B)\left(1-d_B^{-t}\right)\calF_\textup{Haar}^{(1)} + \calO(1/d_A^2).}
    $
\label{first_order_F}
\end{theorem}
The results are derived by representing the typical case of $\{U_t\}$ with Haar expectation values due to self-averaging (see 
Appendix \ref{app:numerical_details}), in the same spirit of Refs.~\cite{cotler2023emergent,ippoliti2023dynamical}.

We numerically (red dots) and experimentally (blue circles) verify Eq.~\eqref{F1_exact} (light red solid) in Fig.~\ref{fig:Fevo}a.  
The experiment on IBM quantum devices, with $N_A=2, N_B=2$ and up to $t=8$ steps, demonstrates a good agreement with the theory up to some additional decay from decoherence noise.

For higher-order frame potentials, we have the following result.
\QZ{Under the commonly-adopted approximation of leading-order Weingarten function~\cite{chan2024projected, ippoliti2022solvable} (see Methods), the typical $K$-th frame potential $\calF^{(K)}(t)$ of the {\em projected ensemble} $\calE_t$ from HDT is lower bounded \QZ{by},
\be
\calF^{(K)}(t) \gtrsim \frac{1}{d_B^t} + \left[1 + \frac{2^K-1}{d_B} + \calO\left(\frac{{\rm exp}(K)}{d_A d_B}\right)\right]\calF_{\rm Haar}^{(K)},
\label{high_order_F}
\ee
in the thermodynamic limit of the system ($d_A \to \infty$).
Here `$\gtrsim$' indicates the approximation involved. To obtain the above results, we map \QZ{the} higher-order frame potential in HDT to the partition function of a statistical model with a chain structure (see Methods). By comparing with Theorem~\ref{first_order_F}, one can immediately realize that Ineq.~\eqref{high_order_F} shares the same decay rate as the exact result of $K=1$. At the same time, the approximation of leading-order Weingarten function requires at least both $d_A\gg1$ and $d_B\gg1$.
}
In Fig.~\ref{fig:Fevo}a, the \QZ{asymptotic lower bound} in \QZ{Ineq.~\eqref{high_order_F}} (red dashed lines) \QZ{surprisingly} aligns with the numerical results (dots) and thus captures the convergence dynamics of higher-order frame potentials. \QZ{The deviation of this lower bound from numerical solution \QZ{is} due to both the approximation error and finite-size effect. Further numerical results in Appendix \ref{app:numerical_details} indicate that Ineq.~\eqref{high_order_F} may indeed be the leading order at $d_A\gg1$ and $d_B\gg1$.}
\QZ{The requirement for Ineq.~\eqref{high_order_F} on the unitary design level is unclear due to the replica technique involved (see Methods), similar to recent studies of frame potential in DT~\cite{ippoliti2022solvable, chan2024projected}.}

We next turn to an efficiently-measurable signature of randomness adopted in the study of quantum device characterization and benchmarking~\cite{boixo2018characterizing, arute2019quantum}---the \emph{probability distribution of probability} (PoP), which describes the distribution of the state \QZ{overlaps} between the ensemble and an arbitrary fixed quantum states $\ket{\phi}$. For \QZ{the} Haar-random ensemble, the PoP is known to follow the Porter-Thomas (PT) distribution~\cite{boixo2018characterizing,porter1956fluctuations},
$
P_{\ket{\psi} \in \text{Haar}}\left[\abs{\braket{\phi|\psi}}^2 = p\right]\propto (1-p)^{d_A - 2}.
$
We experimentally measure the PoP of the \emph{projected ensemble} \QZ{generated via} HDT \QZ{across} various time steps in Fig.~\ref{fig:Fevo}b-d, and compare to \QZ{the} PT distribution (black dashed line). 
At \QZ{the} early stage $t=1$, the measured PoP (light blue) differs significantly from the PT distribution (black dashed) as expected due to limited measurements. \QZ{However}, with iterative evolution and measurements, the histogram converges towards the PT distribution, consistent with the numerical simulation with hardware noise (purple dashed).  
Across subplots b, through c, to d, the agreement with Haar theory improves \QZ{as} $N_A$ increases from $2$, through $4$, to $5$ \QZ{because the finite-system-size} effect \QZ{is} mitigated, as \QZ{indicated by} the noiseless simulation (orange dashed). Compared with previous conventional DT experiment, which is limited to randomizing a data system of dimension $d_A=5$ with a total of $10$ atoms~\cite{choi2023preparing}, the drastic resource reduction of HDT allows the extension to $d_A=32$ dimensional data system ($N_A=5$) with a total of only $8$ qubits.

\section{Space-time trade-off}
Here, we analyze the space-time trade-off in \QZ{detail}, assuming that each unitary $U_t$ is implemented via a gate-based local quantum circuit. We will also consider the \QZ{thermodynamic limit of the data system at $d_A \to \infty$}, although this approximation is not essential. From Theorem~\ref{first_order_F} for $K=1$ and \QZ{Ineq.~\eqref{high_order_F}} for $K\ge 2$, HDT only needs a constant number of \QZ{$N_B \ge K + \log_2(1/\epsilon) + \log_2\left(1-2^{-K}\right)$} ancilla qubits to generate an $\epsilon$-approximate $K$-design state ensemble, given \QZ{a sufficient} number of steps. 
Assuming that \QZ{$N_B \ge K + \log_2(1/\epsilon)$}, due to the fast exponential convergence of $\propto d_B^{-t}$ regardless of order $K$, the necessary number of time steps for convergence in $K$-th order becomes 
\be 
\QZ{N_B \times T \ge K N_A - \log_2(K!) + \log_2\left(1/\epsilon\right) + \calO\left(1/2^{N_A}\right)},
\label{NB_T_tradeoff}
\ee  
leading to the trade-off between space --- the number of ancilla $N_B$ --- and time step $T$, hence the name \QZ{{\em holographic deep thermalization}}. To reach the same performance on $N_A$ qubits, the space-time product $ N_B\times T$ is roughly constant. Note that conventional DT corresponds to $T=1$, which requires an ancilla size \QZ{$N_B\ge KN_A+ \log_2(1/\epsilon)$}, compared to the minimum constant number ancillae \QZ{$N_B \ge K + \log_2(1/\epsilon)$} in the HDT. 

The space-time trade-off can be most easily verified by data collapse when introducing a rescaled time $
\tau 
\QZ{=} t N_B
$ and the rescaled frame potential $f^{(K)}(t) = \calF^{(K)}(t)/\calF^{(K)}(\infty)$. Via the rescaling, the typical frame potential \QZ{asymptotic lower bounds} of HDT \QZ{(Ineq.~\eqref{high_order_F})} with different ancilla size $N_B$ collapse to $ 
\QZ{f^{(K)}(\tau) \ge 1 +  2^{-\tau}\left(1-\calO(1/2^{N_B})\right)/\calF_\text{Haar}^{(K)}}$,  indicated by the black dashed curves in Fig.~\ref{fig:qVol}ab. Indeed, our numerical results (Fig.~\ref{fig:qVol}ab) for various choices of $N_B$ all collapse to the theory prediction with slight deviations due to limited $d_A$, verifying the space-time trade-off. With increasing ancilla size $N_B$ and time steps $t$, we see an alignment in the monotonic decrease of relative frame potential deviation $\delta^{(1)}(T)$ between Haar theory and experiment in Fig.~\ref{fig:qVol}cd. The rescaled frame potential of experiments (blue triangles in Fig.~\ref{fig:qVol}a) also approximately collapses, confirming the universal scaling law.


\begin{figure}[t]
    \centering
    \includegraphics[width=0.45\textwidth]{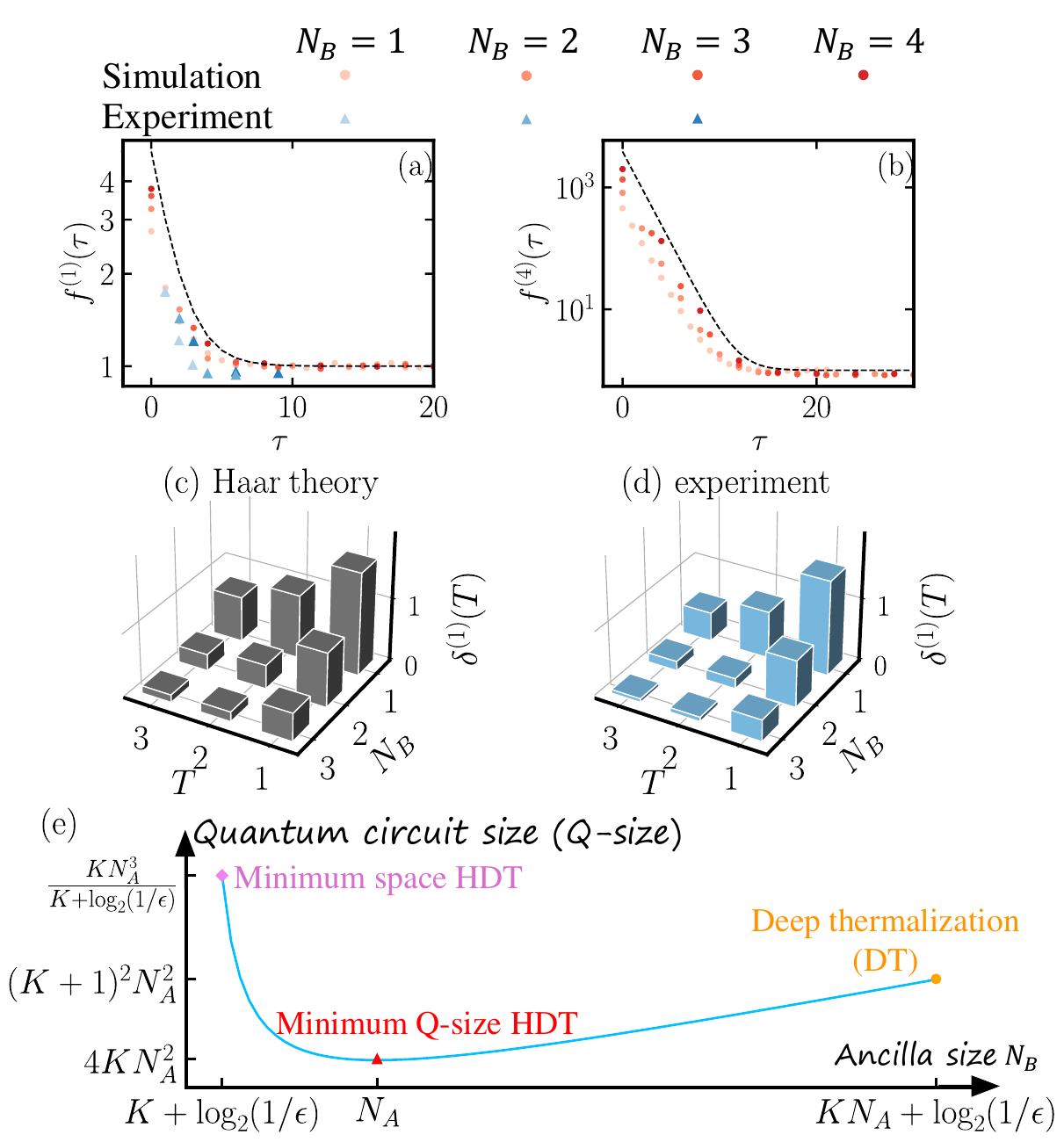}
    \caption{
    \QZ{\bf Space-time trade-off in HDT.} The rescaled frame potential dynamics for (a) $K=1$, $N_A=2$ and (b) $K=4$, $N_A=4$ versus a rescaled time, indicating a universal scaling over different ancilla \QZ{sizes} $N_B$ (light to dark red dots). Black dashed lines represent the theoretical asymptotic universal scaling as $f^{(K)}(\tau) \QZ{=} 1 + 2^{-\tau}/\calF_\text{Haar}^{(K)}$. In the central panel, we plot the relative deviation of frame potential $\delta^{(1)}(T) = \calF^{(1)}(T)/\calF^{(1)}_\text{Haar}(T) - 1$ versus ancilla size $N_B$ and evolution steps $T$ for (c) Haar theory and (d) experimental results. The decrease of the relative deviation $\delta^{(1)}(T)$ reveals the space-time trade-off as $N_B T \gtrsim \rm const$. The corresponding rescaled $f^{(1)}(\tau)$ \QZ{from} experiments are shown by light to dark blue triangles in (a) for $N_B = 1, 2, 3$.
    (e) Quantum circuit size
    of HDT in generating $\epsilon$-approximate $K$-design versus ancilla size for large system $N_A \gg K \log_2(1/\epsilon)$ assuming $K< \log_2(1/\epsilon)$. \QZ{The} orange circle represents conventional DT, which also corresponds to $T=1$ in HDT. \QZ{The} red triangle and violet diamond represent Q-size in HDT with minimum Q-size and minimum space, \QZ{respectively}. The labels on $y$-axis represent the scaling of Q-size from Eq.~\eqref{eq:qvol} \QZ{of Appendix \ref{app:qVol}}.}
    \label{fig:qVol}
\end{figure}

In a gate-based realization of HDT, each unitary $U_t$ requires at least $\propto (N_A+N_B)$ layers of quantum circuits and therefore $\propto (N_A+N_B)^2$ gates in order to represent a fixed typical unitary from \QZ{the approximate Haar random ensemble~\cite{brandao2016local, harrow2023approximate}}. Thus, a more comprehensive way of quantifying the space-time resource is via the quantum circuit size --- the total number of gates --- $\mbox{Q-size}\propto (N_A+N_B)^2 T$. From Ineq.~\eqref{NB_T_tradeoff}, we can perform asymptotic analyses of Q-size, as detailed in 
Appendix \ref{app:qVol}. Fig.~\ref{fig:qVol}e summarizes the trade-off for large systems $N_A \gg K \log_2(1/\epsilon)$ and the common case of $K < \log_2(1/\epsilon)$: the minimum quantum circuit size $\mbox{Q-size}\propto 4KN_A^2$ is achieved at $N_B=N_A$, as indicated by the red triangle. Meanwhile, in conventional DT corresponding to $T=1$ with $N_B = KN_A + \log_2(1/\epsilon)$, we have $\mbox{Q-size} \propto (K+1)^2 N_A^2$, as indicated by the orange dot. For the minimum space (ancilla) case of $N_B \sim K + \log_2(1/\epsilon)$, the quantum circuit size grows as $KN_A^3/(K + \log_2(1/\epsilon))$.

\section{Security of quantum randomness}
An important property of true randomness is that it cannot be predicted prior to generation. In the classical world, prediction of random \QZ{numbers} means that a sequence of classical \QZ{numbers} $\bm X$ is chosen and fixed prior to the generation of random \QZ{numbers} $\bm Y$. In multiple experiments, if the correlation measured by mutual information $I(\bm X:\bm Y)$ is large, then $\bm X$ can \QZ{be used to} infer \QZ{the} random number $\bm Y$. In a lottery scenario, this means one can cheat to win the prize. In the quantum case, we can generalize such a notion by introducing a reference quantum system $R$. Under the condition that no quantum operation can be performed on $R$ after \QZ{the} random number generation process starts, we regard the quantum randomness as insecure if the quantum mutual information between $R$ and the final random quantum system $A_{\rm out}$ is large. For example, if $RA_{\rm out}$ is eventually in the maximally-entangled Bell \QZ{state}, then any measurement on $A_{\rm out}$ and $R$ agrees and the `quantum lottery' is hacked.

\begin{figure}[t]
    \centering
    \includegraphics[width=0.48\textwidth]{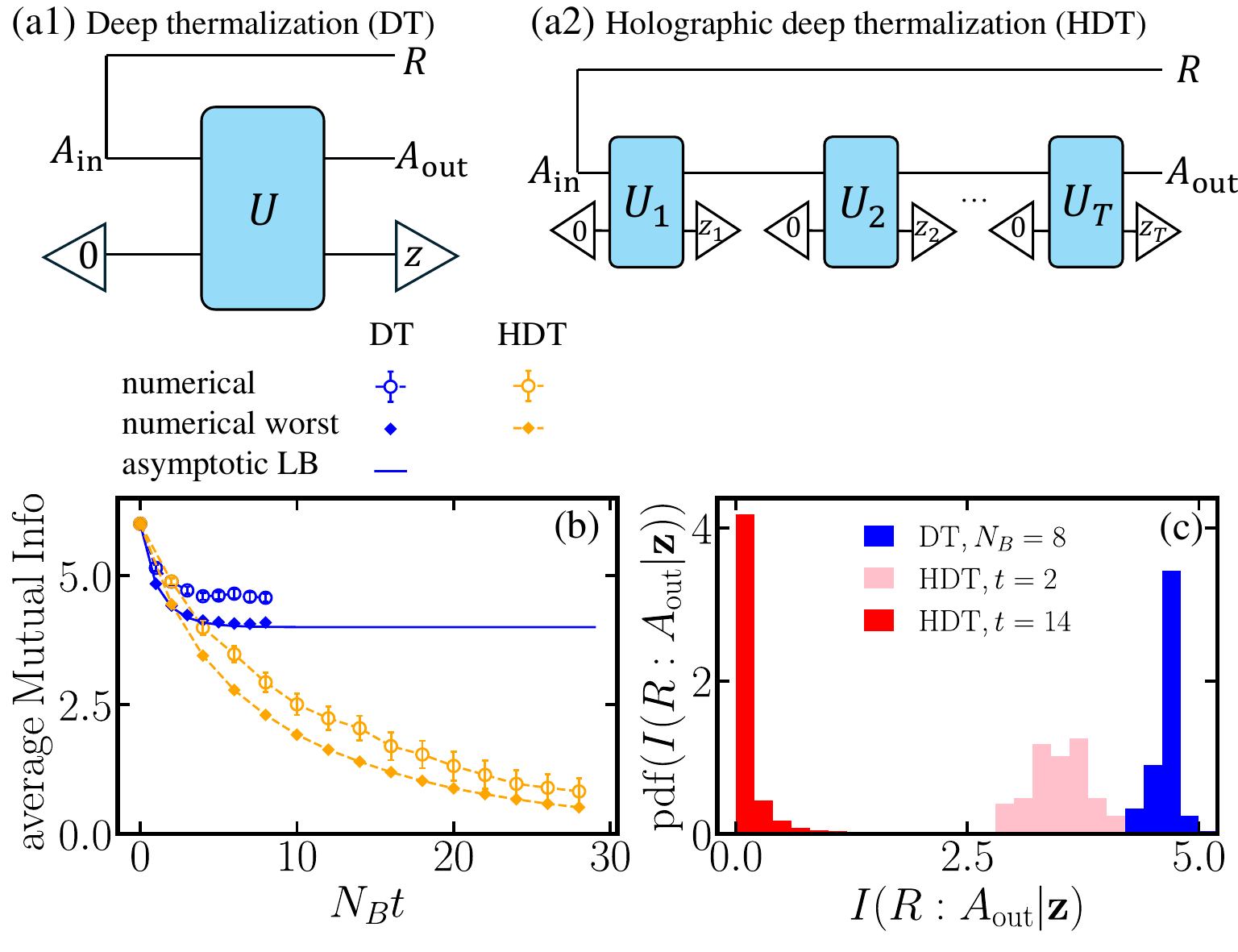}
    \caption{{\bf Dynamics of average mutual information.} 
    \QZ{(a1) and (a2) illustrate the conceptual plot for trajectory-average mutual information in DT and HDT, respectively.}
    The input data system is initially maximally entangled with a reference system $R$. In (b), we plot the dynamics of trajectory-average mutual information (dots) and its R\'enyi lower bound (diamonds) versus $N_B t$ \QZ{for a typical sequence of \QZ{unitaries}}. \QZ{The errorbars in circles represent the fluctuation of mutual information $I(R:A_{\rm out}|\bfz)$ over different conditional \QZ{states}.} The blue solid line is the theory of Eq.~\eqref{eq:avgMI} in Lemma~\ref{lemma:security_DT}. 
    \QZ{In (c), we show the probability density function (pdf) of $I(R:A_{\rm out}|\bfz)$ for DT with large ancilla $N_B = 8$ (blue), HDT at early time $t=2$ (pink) and HDT at late time $t=14$ (red).}
    In DT, we enlarge ancilla system $N_B$ \QZ{while} holding $t=1$, and in HDT, we fix $N_B=2$ but extend evolution steps $t$. In both cases, the data system consists of $N_A = 3$ qubits. \QZ{In the case of DT, $N_Bt$ is limited due to limited numerical resource.}}
    \label{fig:MI}
\end{figure}

We can consider a particular attack based on entanglement engineering. As shown in Fig.~\ref{fig:MI}a, the attack engineers the reference $R$ and the initial data system $A_\text{in}$ in the maximally entangled Bell state. In a DT process (shown in a1), the data system interacts with the large ancilla $B$; while in the HDT process (shown in a2), the data system repetitively interacts with the ancilla via measurement and reset. We evaluate the expected quantum mutual information averaged over the measurement results ${\bfz} = (z_1, \dots, z_T)$, \QZ{$\E_{\bfz} I(R:A_{\rm out}|{\bfz})$, where $I(R:A_{\rm out}|\bfz)$ is the quantum mutual information between $R$ and $A_{\rm out}$ in \QZ{the} conditional state $\ket{\psi_z}_{RA_{\rm out}}$~\cite{gullans2020dynamical}. Note that quantum mutual information of the conditional state is different from the conventional quantum conditional mutual information~\cite{fawzi2015quantum, kato2019quantum, kuwahara2020clustering}. 
One can further show that $\E_{\bfz} I(R:A_{\rm out}|{\bfz}) = \E_{\bfz} 2S(\rho_{A_{\rm out}|\bfz})$ where $S(\cdot)$ is the von Neumann entropy and $\rho_{A_{\rm out}|\bfz} = \tr_{R}(\ketbra{\psi_\bfz}{\psi_\bfz}_{RA_{\rm out}})$ is the reduced state on subsystem $A_{\rm out}$ of conditional output state $\ket{\psi_z}_{R A_{\rm out}}$. } 
When further averaged over the Haar ensemble to represent the typical case \QZ{of unitary implementation}, we obtain the following results, with \QZ{proofs} in Appendix \ref{app:mutual_info}. 
\begin{lemma}
\label{lemma:security_DT}
    The typical quantum mutual information averaged over the measurement results in deep thermalization only \QZ{decreases} by at most $2$, regardless of ancilla size. More precisely, \QZ{in the thermodynamic limit of $d_A \to \infty$, it is} lower bounded by,
    \be\label{eq:avgMI}
    \begin{aligned}        
        &
        \E_{\bfz}[I(R:A_\textup{out}|{\bfz})]\ge 
        \E_{\bfz} [2S_2(\QZ{\rho_{A_{\rm out}|\bfz}})]\\
        &\ge \QZ{2(N_A-1) - 2\log_2\left(1- \frac{1}{2d_B} +\calO\left(\frac{1}{2d_A^2}\right)\right)},
    \end{aligned}
    \ee
    where $S_2(\QZ{\rho_{A_{\rm out}|\bfz}})$ is the R\'enyi-2 entropy of reduced state $\tr_{\QZ{R}}(\ketbra{\psi_{\bfz}}{\psi_{\bfz}}_{RA_{\rm out}})$ conditioned on measurement result $\bfz$. 
\end{lemma}

In Fig.~\ref{fig:MI}b, the numerical results of average mutual information $ 
\E_{\bfz}[I(R:A_\textup{out}|{\bfz})]$ (dots) decay synchronously with the R\'enyi lower bound $ 
\E_{\bfz} [2S_2(\QZ{\rho_{R|\bfz}})]$ (diamonds) in DT (blue) and HDT (orange), indicating a well-behaved numerically-evaluated lower bound. 
Moreover, the numerical results of DT lower bound (blue diamonds) agree with the analytical lower bound (blue solid line) in Ineq.~\eqref{eq:avgMI}, which only decays by less than 2. On the contrary, HDT allows the continuous decrease of mutual information (orange dots and diamonds) towards zero, as the number of time steps increases, converging towards secure quantum randomness independent of the reference system. \QZ{In both DT and HDT, the fluctuation (errorbar) of mutual information over different conditional state is relatively small, indicating that the average mutual information can represent the typical case being sampled. Moreover, Fig.~\ref{fig:MI}c depicts the distribution of the mutual information for \QZ{the} conditional state $\ket{\psi_\bfz}_{RA_{\rm out}}$ of the \emph{projected ensemble} generated by DT and HDT. With increasing number of effective ancilla qubits $N_B t$ in both DT and HDT, the mutual information for \QZ{the} conditional state in \QZ{the} \emph{projected ensemble} becomes highly concentrated (\QZ{red and blue}), therefore the mutual information for a typical sampled state is well represented by the defined average quantum mutual information $\E_z I(R:A_{\rm out}|\bfz)$. In fact, we expect that the probability to sample any non-typical conditional state would be exponentially small $\sim 2^{-N_B t}$.}

\section{Discussion}
In conclusion, we have proposed HDT for \QZ{the generation of} genuine quantum random states, and have shown the convergence dynamics of output states \QZ{toward the} Haar ensemble, and its robustness against \QZ{the} entanglement attack in contrast to conventional DT. Moreover, we have shown the \QZ{emergence of} space-time trade-off in HDT and identified the optimal configuration meeting \QZ{the} minimum resource requirement.
The proposed HDT protocol can be further enhanced by generative quantum machine learning. 
As we detail in the Methods, parameterizing and training the unitaries in each step enables a constant-factor \QZ{improvement in} convergence. Our model can also be regarded as a generalized monitored circuit with fixed measurement rate and location, which provides insights \QZ{into} the understanding of complexity growth in the monitored circuit beyond entanglement growth transition~\cite{skinner2019measurement}.

We point out some future directions. The \QZ{exact solution of the} higher-order frame potential \QZ{remains an open problem for future theoretical investigation}. \QZ{Moreover, in order to generate $K$-design states, the required level of design in the unitary of HDT is an important open problem, to enable practical deep thermalization. The approaches in Ref.~\cite{cotler2023emergent} may be beneficial towards this goal.} It is also an open direction to explore \QZ{a} Hamiltonian-dynamics\QZ{-based} version of HDT, where \QZ{the effective temperature of the quantum system can be either infinite or finite~\cite{mark2024maximum}, and the system Hamiltonian \QZ{may exhibit} constraints} (e.g., symmetry\QZ{~\cite{varikuti2024unraveling, chang2024deep}} or quantum many-body scar~\cite{bhore2023deep}).

Another intriguing direction is to establish the rigorous notion of security in quantum random state generation. While quantum random number generation \QZ{is primarily concerned with} the generation of classical random numbers from quantum \QZ{devices}~\cite{vazirani2012certifiable,cao2016source,drahi2020certified}, the security of \QZ{the} random quantum state generation \QZ{has been} less explored. The study of computational pseudo-random quantum \QZ{states focuses} on the closeness to Haar~\cite{ji2018pseudorandom,ananth2022cryptography}, while in our case we \QZ{focus on} the device-level attack. The information security paradigm in DT and HDT can also be explored as an extension of \QZ{the} Hayden-Preskill protocol~\cite{hayden2007black} to study the intersection between quantum measurements and \QZ{the} black hole information paradox.

\section{Methods}

\subsection{Preliminary}
Statistical properties of a quantum state ensemble $\calE(\Psi)=\{\ket{\psi_z}, p_z\}$ are characterized by its $K$-th moment ($K\ge 1$), which is captured by the $K$-th moment operator of \QZ{the} ensemble~\cite{roberts2017chaos, cotler2023emergent, ho2022exact, ippoliti2023dynamical},
\begin{align}
    \rho^{(K)} = \sum_z p_z\left(\ketbra{\psi_z}{\psi_z}\right)^{\otimes K},
    \label{def_rho_k}
\end{align}
Specifically, for \QZ{the} Haar ensemble in Hilbert space of dimension $d_A$, the moment operator 
\begin{align}
    \rho_\text{Haar}^{(K)}    &= \frac{\sum_{\pi \in S_K} \hat{\pi}}{\prod_{i=0}^{K-1} (d_A+i)} =\Pi_K\binom{K+d_A-1}{d_A-1}^{-1},
    \label{eq:haar_mmt}
\end{align}
where $S_K$ is the symmetric group on a $K$-element set, $\hat{\pi}$ is the operator representation of \QZ{a} permutation on $\calH^{\otimes K}$, and $\Pi_K = \sum_{\pi \in S_K} \hat{\pi}/K!$ is the projector onto the symmetric subspace.

In the main text, we adopted the frame potential to quantify the closeness to \QZ{the} Haar state ensemble, which \QZ{is} also \QZ{equal to} the purity of the moment operator, 
$
\calF^{(K)}=\tr[\rho^{(K)}{}^2].
$
One can show that the relative deviation in frame potential is equivalent to the normalized moment distance with $\|\cdot \|_2$ denoting Schatten-$2$ norm \QZ{(see Appendix \ref{app:preliminary} for more details)}
\be 
\sqrt{\delta^{(K)}} = \norm{\rho^{(K)}-\rho_\text{Haar}^{(K)}}_2/\norm{\rho_\text{Haar}^{(K)}}_2.
\ee 
\QZ{As a side note, Refs.~\cite{ho2022exact,cotler2023emergent} adopt \QZ{the} one-norm to define the deviation; while other works adopts the two-norm~\cite{ippoliti2022solvable,ippoliti2023dynamical}.}

\subsection{Statistical model mapping and proof sketch}
\QZ{In this section, we detail the mapping from the frame potential in HDT to a statistical model and provide a proof sketch \QZ{of} the Theorem~\ref{first_order_F} and Ineq.~\eqref{high_order_F}. The full derivations can be found in Appendix \ref{app:F_dynamics}.}

\QZ{Utilizing the equivalent circuit of HDT with expanded bath in Fig.~\ref{fig:stat_model} (top), the post-measurement state conditioned on \QZ{the} measurement outcome $\bfz = (z_1, \dots, z_t)$ at \QZ{the} $t$-th step can be expressed as
\be
    \ket{\psi_\bfz} = {}_{\bm B}\braket{\bfz|U_t\cdots U_1|0}_{A\bm B}/\sqrt{p_{\bfz}} \equiv \ket{\tilde{\psi}_{\bfz}}/\sqrt{p_\bfz},
\ee
where $\bm B = B_1\cdots B_t$ \QZ{refers to} the \QZ{combined} system of expanded baths, and $p_\bfz = \braket{\tilde\psi_\bfz|\tilde\psi_\bfz}$ is the corresponding probability. Here we \QZ{define} $\ket{\tilde{\psi}_\bfz}$ as the unnormalized conditional state for convenience. In this way, the $K$-th frame potential of \emph{projected ensemble} $\calE_t = \{\ket{\psi_\bfz}, p_\bfz\}$ defined in Eq.~\eqref{eq:fp_def} can be written as
\be
    \calF^{(K)}(t) = \sum_{\bfz, \bfz'} p_\bfz^{1-K} p_{\bfz'}^{1-K} |\braket{\tilde{\psi}_\bfz|\tilde{\psi}_{\bfz'}}|^{2K}.
    \label{eq:Fk_method}
\ee
To evaluate the Haar ensemble average result, we rely on the known Haar unitary twirling of \QZ{an} operator in \QZ{the} $n$-\QZ{fold tensor product of the} Hilbert space with dimension $d$~\cite{roberts2017chaos}
\be
    \int_{\rm Haar} dU U^{\otimes n} (\cdot)U^\dagger{}^{\otimes n} = \sum_{\sigma, \pi \in S_n} {\rm Wg}_d(\sigma^{-1}\pi,n) \hat{\sigma} \tr(\hat{\pi}^\dagger \cdot),
    \label{eq:haar_twirling}
\ee
where ${\rm Wg}_d(\sigma^{-1}\pi,n)$ is the Weingarten function, and $S_n$ is the permutation group of $n$ elements.
}

\begin{figure}[t]
    \centering
    \includegraphics[width=0.45\textwidth]{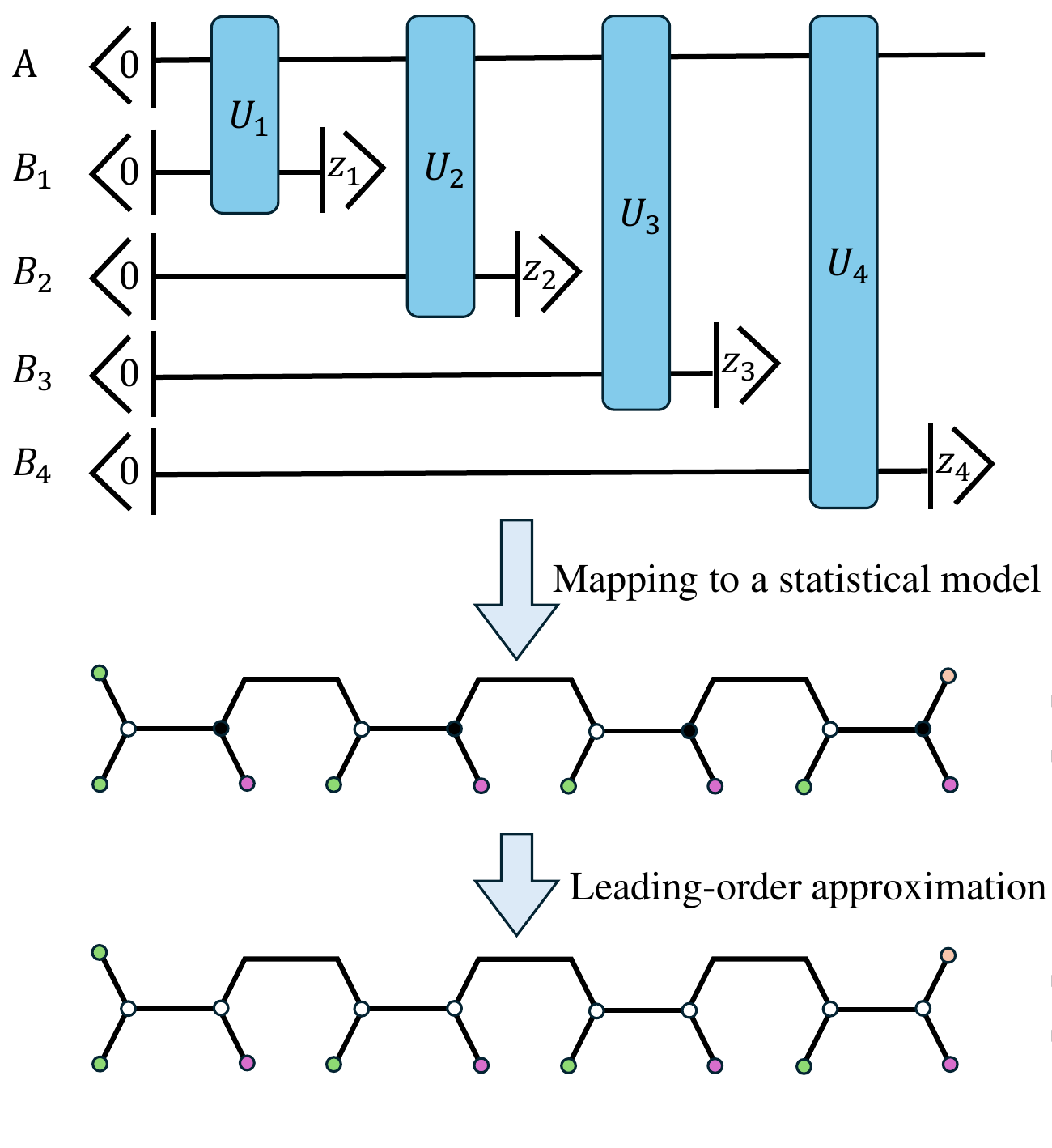}
    \caption{\QZ{{\bf Mapping from HDT to a spin statistical model.} The pseudo frame potential of the \emph{projected ensemble} can be formulated as the partition function of this statistical model (center). White and black vertices are independent spins representing permutations. Green, orange and pink vertices are boundary spins representing $\ketbra{\bm 0}{\bm 0}^{\otimes 2m}, \hat{e}, Q_B$ in Eq.~\eqref{eq:Fk_stat} correspondingly. Horizontal lines connecting `open' hexagons are Weingarten functions and other lines are inner products between operators. In the bottom, we show the statistical model under the leading-order approximation in Weingarten function in Eq.~\eqref{eq:W_approx} with white vertices only. Here we show an example of $t=4$.}}
    \label{fig:stat_model}
\end{figure}

\QZ{
{\em Proof sketch of Theorem~\ref{first_order_F}.---} The first order frame potential $\calF^{(1)}(t)$ only involves two copies of unitaries and its Hermitian adjoints (see Eq.~\eqref{eq:Fk_method}), we therefore only need to keep track of the unitary twirling results of permutations in $S_2$---identity and swap. \QZ{The evolution of} permutations in $S_2$ through the HDT \QZ{framework yields} Theorem~\ref{first_order_F}.
}

\QZ{For higher order frame potential $\calF^{(K)}(t)$ with $K > 1$, since Eq.~\eqref{eq:Fk_method} includes non-positive exponents $p_\bfz^{1-K}, p_{\bfz'}^{1-K}$,
we need to utilize the replica trick~\cite{ippoliti2022solvable, chan2024projected} and define a pseudo frame potential as
\begin{align}
    &\calF^{(n,K)} \equiv \sum_{\bfz, \bfz'} p_\bfz^{n} p_{\bfz'}^{n} |\braket{\tilde{\psi}_\bfz|\tilde{\psi}_{\bfz'}}|^{2K}\\
    &= \tr\left(\left(U_t \cdots U_1\right)^{\otimes 2m}\ketbra{\bm 0}{\bm 0}_{AB}^{\otimes 2m}\left(U_1^\dagger \cdots U_t^\dagger\right)^{\otimes 2m} Q_B\right),
    \label{eq:pseudo_fp_def}
\end{align}
where we define $m \equiv n + K$ for simplicity and $Q_B \equiv \sum_{\bm z, \bm z'}\ketbra{\bm z^{\otimes n}, \bm z'^{\otimes K}, \bm z^{\otimes K}, \bm z'^{\otimes n}}{\bm z^{\otimes n}, \bm z^{\otimes K}, \bm z'^{\otimes K}, \bm z'^{\otimes n}}$ is the boundary condition operator acting on the expanded bath system $\bm B$. 
Applying the Haar twirling result in Eq.~\eqref{eq:haar_twirling} to the sequence of unitaries from $U_t$ to $U_1$ in the pseudo frame potential above, we can express the pseudo frame potential in terms of a summation over permutations as
\begin{align}
    &\calF^{(n,K)}(t)=\nonumber\\
    &\sum_{\substack{\bm \sigma \in S_{2m}\\
    \bm \pi \in S_{2m}}} \left[\prod_{i=1}^t {\rm Wg}_d\left(\sigma_i^{-1}\pi_i,2m\right) \tr_B\left(\hat{\pi}_i^\dagger Q_B\right) \tr_B\left(\hat{\pi}_i^\dagger \ketbra{\bm 0}{\bm 0}\right)\right] \nonumber\\
    & \qquad \quad \times \tr_A\left(\hat{\pi}_t^\dagger \ketbra{\bm 0}{\bm 0}\right)  \left[\prod_{i=2}^t \tr_A\left(\hat{\pi}_i^\dagger \hat{\sigma}_{i-1}\right)\right]  \tr_A\left(\hat{\pi}_1^\dagger \hat{e}\right),
    \label{eq:Fk_stat}
\end{align}
where $\bm \sigma = (\sigma_1, \sigma_2, \dots, \sigma_t), \bm \pi = (\pi_1, \pi_2, \dots \pi_t)$ are \QZ{sequence of permutations, each independently drawn from the permutation group $S_{2m}$}. Here $\hat{e}$ is the identity operation on $2m$ replicas in Hilbert space with dimension $d = d_A d_B$.
}

\QZ{
Thanks to the studies of emergent statistical mechanics in random circuits~\cite{zhou2019emergent, hunter2019unitary}, we can interpret Eq.~\eqref{eq:Fk_stat} as a `spin' statistical model shown in Fig.~\ref{fig:stat_model} (center). 
The spin away from boundaries (black or white vertex) in this statistical model represents the independent permutations from unitary twirling.
The boundary spin represents the initial input state $\ketbra{0}{0}^{\otimes (2m)}$ (green), $\hat{e}$ (orange) and $Q_B$ (magenta) in Eq.~\eqref{eq:Fk_stat}. 
Black lines are associated weights in the coupling of spins. Specifically, the horizontal lines are Weingarten functions and other black inclined lines or polylines denote the inner product between operators. Therefore, the pseudo frame potential in  Eq.~\eqref{eq:Fk_stat} in HDT can be interpreted as the partition function of this spin statistical model, where the summation over spin configuration are indeed the summation over permutation group $S_{2m}$.
Due to the sequential circuit structure of HDT, the statistical model of HDT formulates a chain structure, instead of the hexagon lattice with spatial translational symmetry in the regular random circuits~\cite{zhou2019emergent, hunter2019unitary}.
}

\QZ{
It is overall challenging to derive the exact expression for the pseudo frame potential of HDT due to the iterative interactions between spins (permutations) in the statistical model presented in Eq.~\eqref{eq:Fk_stat}, which are absent in the case of DT. To obtain analytical understanding to the pseudo frame potential, we first adopt the leading-order approximation of Weingarten function in thermodynamic limit $d \to \infty$~\cite{chan2024projected, ippoliti2022solvable} as
\be
    {\rm Wg}_d(\sigma^{-1}\pi, 2m) = \frac{\delta_{\sigma, \pi}}{d^{2m}}\left(1 + o(1)\right),
    \label{eq:W_approx}
\ee
which simplifies the statistical model presented in Fig.~\ref{fig:stat_model} (bottom). Here every pair of white vertices connected by the direct short line are forced to be identical, while vertices from different pairs remain independent. Under this approximation, we have
\begin{align}
\calF^{(n,K)}(t) &\simeq \calF^{(n,K)}_{\rm approx}(t)
\\
&\equiv \sum_{\bm \sigma \in S_{2m}} \left[\prod_{i=1}^t {\rm Wg}_d\left(e,2m\right) \tr_B\left(\hat{\sigma}_i^\dagger Q\right)\right] \nonumber\\
        &\qquad \quad \times \left[\prod_{i=2}^t \tr_A\left(\hat{\sigma}_i^\dagger \hat{\sigma}_{i-1}\right)\right]  \tr_A\left(\hat{\sigma}_1^\dagger\right).
        \label{F_approx}
\end{align}
We can further analytically derive a lower bound as follows. 
\begin{lemma}
    \label{lem:F_pseudo_LB}
    In the thermodynamic limit of data system ($d_A \to \infty$),
    \begin{align}
        &\calF^{(n,K)}_{\rm approx}(t)
         \ge \frac{\left(\prod_{i=0}^{m-1}(d_A+i)\right)^{2t}}{d_A^{2t}}  \nonumber\\ &\quad \times\left[\frac{1}{d_B^t} + \left(1 + \frac{2^K-1}{d_B} + \calO\left(\frac{{\rm exp}(K)}{d_A d_B}\right)\right)\calF_{\rm Haar}^{(K)}\right].
    \end{align}
\end{lemma}
{\em Proof sketch of Lemma~\ref{lem:F_pseudo_LB}.---} We derive this lemma by considering two sets of permutations sequences from the summation of Eq.~\eqref{F_approx}, since all summation terms are positive. Firstly, all spins are limited within the set of disjoint permutations $G_0 = S_{m}\times S_{m}$ where there does not exist any permutation across the first and last $m$ replicas, such as identity operation $e$ from the boundary operator. This type of permutations $G_0$ contributes to the dynamical term in Lemma~\ref{lem:F_pseudo_LB}. Secondly, all spins except for the right boundary $\sigma_1$ are restricted within a set $G^*$ to maximize its inner product with $Q$ while $\sigma_1$ is allowed to be other permutations as well, which can be interpreted as a 1-domain wall picture and corresponds to the converged value in 
Lemma~\ref{lem:F_pseudo_LB}. 
}

\QZ{
By taking the replica trick $m \to 1$, we can reduce Lemma~\ref{lem:F_pseudo_LB} to Ineq.~\eqref{high_order_F} in the main text.
We leave the full characterization of the partition function of the statistical model thus the frame potential of HDT as an interesting future direction. 
}

\subsection{Details of experiments}
We validate the \emph{projected ensemble} of HDT through experiments conducted on IBM Quantum devices~\cite{Qiskit}. Our implementation utilizes the mid-circuit measurement and reset \QZ{operations}.

In Fig.~\ref{fig:Fevo}a, we implement the hardware-efficient ansatz (HEA)~\cite{kandala2017hardware} with $L=4$ layers on a system \QZ{with} $N_A=2, N_B=2$ qubits to generate sufficiently random and complex unitaries at each time step of HDT. To evaluate the first-order frame potential, we \QZ{consider} the purity of the average state of the \emph{projected ensemble} on the data system \QZ{as the purity is equal to the first-order frame potential by definition. The purity of average state can be further evaluated by}
\begin{equation}\label{eq:purity}
    \tr(\bigl(\rho^{(1)}\bigr)^2) = \frac{1}{d}\sum_{P \in \calP_{N_A}} \tr(\rho^{(1)} P)^2,
\end{equation}
where $\calP_{N_A}$ is the Pauli group on $N_A$ qubits omitting phases, and $\rho^{(1)}$ is the first-order moment operator defined in Eq.~\eqref{def_rho_k}. Each blue dot in Fig.~\ref{fig:Fevo}a represents an average over $100$ sets of randomly initialized HEA parameters, implemented on IBM Quantum Sherbrooke. For each Pauli expectation in Eq.~\eqref{eq:purity}, we take $4096$ measurement shots for an accurate estimation. For \QZ{a} future extension to large data systems, the classical shadow~\cite{huang2020predicting} provides an alternative efficient method to estimate purity. 

In Fig.~\ref{fig:Fevo}b-d, the HEA is employed again to compute the PoP of \emph{projected ensembles} with $L=N_A+N_B$ layers of circuits in each step. Unlike the frame potential experiment, a single set of circuit gate parameters is randomly chosen and kept fixed across repetitions. For each experiment, $10^5$ measurement shots are taken to gather a sufficient number of measurement strings from the joint system. \QZ{Each experiment is}  repeated $10$ times on IBM Quantum Brisbane to obtain the average PoP at different time steps in various systems. Noiseless and noisy simulations are performed using the same setup, with ideal and noisy circuit models provided by IBM Quantum Qiskit. For \QZ{an} $N_A=2$ qubit system with dimension $d_A=4$, we choose $N_B=2$ ancilla qubits and implement six steps of HDT, each step with a circuit depth of 4. For \QZ{an} $N_A=4$ qubit system ($d_A=16$), we still choose $N_B=2$ ancilla qubits and implement six steps of HDT, each step with a circuit depth of 6. We also extend the experiment to \QZ{a system of} $N_A=5$ qubits ($d_A=32$) with $N_B=3$ ancilla and three steps, each with a circuit depth of \QZ{$L=8$}. 

In Fig.~\ref{fig:qVol}d, we follow the same design as in Fig.~\ref{fig:Fevo}a on a system of $N_A=2$ data qubits but \QZ{varying the number of ancilla qubits $N_B$} and time steps $T$. The circuit depth is chosen to be $L=2 + N_B$ layers. Each blue bar is averaged over $10$ sets of randomly initialized HEA circuits implemented on IBM Quantum Brisbane.

In the end, we list some of the median calibration data of IBM Quantum Sherbrooke and Brisbane in \autoref{tab:device_err}.

\begin{table}[H]
\centering
\begin{tabular}{|c|c c|} 
\hline
Error & Sherbrooke & Brisbane \\
\hline\hline
2-qubit Gate Error & \num{7.536e-3} & \num{8.006e-3}\\
\hline
Readout Error & 0.0133 & 0.0154 \\
\hline
$T_1$ & \SI{272.47}{\us} & \SI{228.99}{\us}\\
\hline
$T_2$ & \SI{175.87}{\us} & \SI{143.04}{\us}\\
\hline
\end{tabular}
\caption{Median calibration data of the IBM \QZ{Quantum} devices.\label{tab:device_err}}
\end{table}

\subsection{Quantum machine learning enhancement}
The task of random quantum state generation can be \QZ{viewed} as a quantum transport problem from a Dirac-$\delta$ distributed ensemble of a single initial state towards the Haar ensemble of states. In this regard, we can further enhance HDT via generative quantum machine learning (QML)~\cite{zhang2024generative}.
As shown in Fig.~\ref{fig:scheme_train}a, we parameterize \QZ{unitaries as $U(\bm \theta_1), \dots, U(\bm \theta_T)$} and optimize the parameters $\{\bm \theta_t\}_{t=1}^T$ to obtain the best \QZ{approximation to the} Haar-random ensemble at the output. 

Towards this goal, we consider a divide-and-conquer strategy of only training the unitary $U(\bm \theta_t)$ in the $t$-th step \QZ{by} minimizing a loss function $\calL_t(\calE_t)$ on the state ensemble $\calE_t$. While a greedy strategy may hope for $\calE_t$ to be close to \QZ{the} Haar-random ensemble $\calE_\text{Haar}$ at each training step, we \QZ{instead adopt} a less greedy cost function $\calL_t(\calE_t)=(1-q_t)D(\calE_t, \ket{\psi_0})+q_t D(\calE_t, \calE_\text{Haar})$ as an interpolation of `distances' to the initial single-state $\ket{\psi_0}$ and Haar ensemble, with a monotonically increasing scheduling function $q_t\in [0,1]$. Here $D(\cdot,\cdot)$ \QZ{characterizes} the `distance' between two ensembles (see 
Appendix \ref{app:QML_details}). To enable efficient evaluation of the loss function on actual quantum devices, we consider the simplest choice of `distance' $D(\cdot,\cdot)$, the super-fidelity between the average state $\rho_t=\E_{\phi\sim \calE_t}\state{\phi}$ of the ensembles---only \QZ{tracking} the first-moment information. We expect that the final generated state ensemble $\calE_T$ \QZ{can} still capture some of the higher-moment properties of true Haar ensemble through a large enough number of steps $T$ thus sufficient randomness from measurements. 

\begin{figure}[t]
    \centering
    \includegraphics[width=0.45\textwidth]{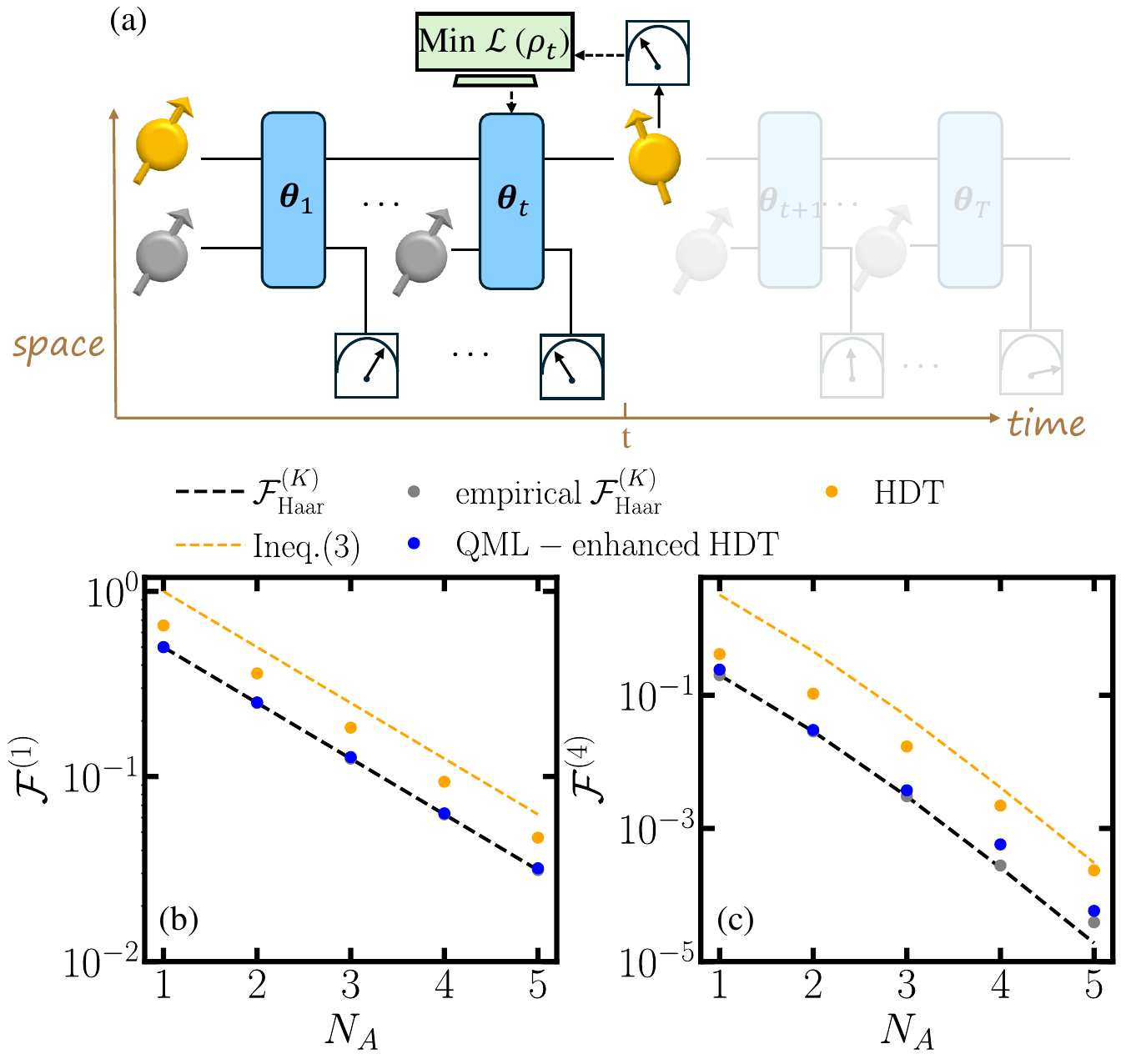}
    \caption{\QZ{\bf Quantum-machine-learning-enhanced HDT.} In (a), we show the conceptual plot of the \QZ{QML-enhanced HDT}. In (b) and (c), we show the frame potential for $K=1$ and $K=4$ of output state ensemble. Orange dots represent \QZ{the} converged frame potential in HDT with random unitary implementation. Orange dashed line is the theoretical prediction in Ineq.~\eqref{high_order_F}. Blue dots are \QZ{the} frame \QZ{potentials} of output state \QZ{ensembles} from QML-enhanced HDT. Black dashed line and gray dots show the exact and empirical \QZ{value} of \QZ{the} Haar ensemble frame potential. }
    \label{fig:scheme_train}
\end{figure}

We numerically simulate the training of the QML-enhanced HDT with $T=32$ and $N_B = 1$ \QZ{ancilla qubits}. The results, \QZ{presented} in Fig.~\ref{fig:scheme_train}b\QZ{-}c, \QZ{shows that} QML-enhanced frame potential (blue dots) \QZ{closely matches} the exact Haar value (black dashed \QZ{line}), \QZ{with minor} deviations \QZ{attributed to} finite \QZ{sampling} indicated by the finite-sampling Haar results (gray dots). \QZ{A constant improvement in performance} over the untrained HDT is \QZ{observed with} the QML enhancement. Details of the quantum circuit ansatz and training are provided in 
Appendix \ref{app:QML_details}.

\subsection{Numerical simulation details}
We numerically simulate the HDT without \QZ{and with} QML enhancement \QZ{using} \texttt{TensorCircuit}~\cite{zhang2023tensorcircuit}.

In Fig.~\ref{fig:Fevo}a and Fig.~\ref{fig:qVol}a\QZ{-}b, each \QZ{data point} is an average over $20$ random Haar unitary realizations, where each ensemble consists of $5\times 10^4$ samples. 

In Fig.~\ref{fig:MI}b, the average mutual information and \QZ{the} corresponding R\'enyi lower bound for both \QZ{DT} and HDT are evaluated exactly over all possible measurement \QZ{outcomes when} $t \le 8$, meanwhile for $t>8$, we take Monte Carlo sampling over $2\times 10^4$ measurement results. We further average over $50$ randomly sampled Haar unitary realizations.

In Fig.~\ref{fig:scheme_train}, for \QZ{HDT} with \QZ{Haar} random unitary implementation, we average over $50$ Haar unitaries realizations with \QZ{each} ensemble \QZ{containing} $5\times 10^4$ samples. For QML-enhanced holographic deep thermalization, we average over $20$ post-measurement ensembles, each of which includes $5\times 10^4$ samples. The empirical $\calF^{(K)}_\text{Haar}$ is \QZ{estimated by averaging} over $50$ ensembles of $5\times 10^4$ Haar-random states.


%



\


\ 

{\noindent\bfseries Acknowledgments.} BZ thanks Matteo Ippoliti for discussions, BZ and PX thank Runzhe Mo for suggestions on the experiment. QZ and BZ acknowledge support from NSF (CCF-2240641, OMA-2326746, 2350153), ONR N00014-23-1-2296, AFOSR MURI FA9550-24-1-0349 and DARPA (HR0011-24-9-0362, HR00112490453, D24AC00153-02). XC and PX acknowledge support from NSF (DMS-2413404, 2347760). This work was partially funded by an unrestricted gift from Google and a gift from the Simons Foundation. The experiment was conducted using IBM Quantum Systems provided through USC’s IBM Quantum Innovation Center. 






\newpage 

\appendix

\begin{widetext}

\section{Preliminary}
\label{app:preliminary}
\subsection{Deviation to Haar state ensemble}

To quantify the approximation of an arbitrary state ensemble $\calE$ to the Haar ensemble (a $K$-design), we consider the (normalized) $p$-order moment distance and the relative deviation of frame potential as~\cite{ippoliti2023dynamical}
\begin{align}
    \Delta_{\calE,p}^{(K)} &= \frac{\norm{\rho_\calE^{(K)} - \rho_\text{Haar}^{(K)}}_p}{\norm{\rho_\text{Haar}^{(K)}}_p},\\
    \delta_\calE^{(K)} &= \frac{\calF_\calE^{(K)}}{\calF_\text{Haar}^{(K)}} - 1,
\end{align}
where $\|\cdot \|_p$ is the $p$-Schatten norm and $\rho_\calE^{(K)}$ is defined in Eq.~\eqref{def_rho_k} of the main text. Commonly used orders of norms include $p=1$ (trace norm), $p=2$ (Frobenius norm) and $p=\infty$ (operator norm). Note that $\norm{\rho_\text{Haar}^{(K)}}_1 = 1$, then
\be
    \Delta_{\calE, 1}^{(K)} = \norm{\rho_\calE^{(K)} - \rho_\text{Haar}^{(K)}}_1 = 2 D_{\rm tr}\left(\rho_\calE^{(K)}, \rho_\text{Haar}^{(K)}\right),
\ee
which is equivalent to the trace distance between $K$-th moment of $\calE$ and the Haar ensemble, widely explored in Refs.~\cite{cotler2023emergent,ho2022exact}. 
We can also show that for $p=2$, it corresponds to the frame potential as~\cite{ippoliti2023dynamical}
\begin{align}
    \Delta_{\calE, 2}^{(K)} &= \frac{\norm{\rho_\calE^{(K)} - \rho_\text{Haar}^{(K)}}_2}{\norm{\rho_\text{Haar}^{(K)}}_2}\\
    &= \frac{\tr[\left(\rho_\calE^{(K)} - \rho_\text{Haar}^{(K)}\right)^2]^{1/2}}{\tr[\left(\rho_\text{Haar}^{(K)}\right)^2]^{1/2}}\\
    &= \left(\frac{\calF_{\calE}^{(K)} - \calF_\text{Haar}^{(K)} }{\calF_\text{Haar}^{(K)}}\right)^{1/2} = \sqrt{\delta_\calE^{(K)}},
\end{align}
where we utilize $\tr(\rho_\calE^{(K)} \rho_\text{Haar}^{(K)}) = \calF^{(K)}_\text{Haar}$ in the third line. The norm with $p=1$ can also be bounded by $p=2$ norm with norm inequalities, therefore we mainly focus on the frame potential in the main text and in the following calculation.

\subsection{Important identities}

In this subsection, we introduce some important identities used to derive the frame potential in {\em deep thermalization} (DT) and {\em holographic deep thermalization} (HDT).

The $K$-th moment of the Haar state ensemble, $\rho_{\rm Haar}^{(K)}$,
\be
    \rho_{\rm Haar}^{(K)} \equiv \int_{\rm Haar} d\psi \ketbra{\psi}{\psi}^{\otimes K} = \frac{(d-1)!}{(K+d-1)!}\sum_{\pi \in S_K} \hat{\pi},
    \label{eq:K_moment_haar}
\ee
where $d$ is the dimension of Hilbert space, $S_K$ is the permutation group of $K$ elements and $\hat{\pi}$ is the permutation operator on the $K$-folded Hilbert space as
$\hat{\sigma}\ket{\psi_1, \dots, \psi_K} = \ket{\psi_{\sigma(1)}, \dots, \psi_{\sigma(K)}}$. This identity is also referred in Methods.

Next, we introduce the doubled Hilbert space representation of operators for notation simplicity. For any operator $\hat{A}$, we define an equivalent representation in the doubled Hilbert space as
\be
    \hat{A} = \sum_{ij} \braket{i|A|j} \ketbra{i}{j}  \to |\hat{A}) \equiv \sum_{ij} \braket{i|A|j} \ket{i} \ket{j}.
\ee
Therefore, the operator $A$ evolved by a unitary $U$ can be equivalently written as
\be
    U \hat{A} U^\dagger \to \left[U\otimes U^*\right]|\hat{A}), 
    \label{eq:twirling_dual}
\ee
where $U^*$ denotes the complex conjugate of $U$. With this notation, one can easily see that $(A|B) = \tr(A^\dagger B)$. 

With the doubled Hilbert space representation in hand, we introduce the Haar unitary twirling identity
\begin{align}
    \E_{\rm Haar}\left[\left(U \otimes U^*\right)^{\otimes n}\right] &= \sum_{\sigma, \pi \in S_n} {\rm Wg}(\sigma^{-1}\pi, n;d) |\hat{\sigma})(\hat{\pi}| \label{eq:haar_integral}\\
    &\simeq {\rm Wg}(e, n; d) \sum_{\sigma \in S_n} |\hat{\sigma})(\hat{\sigma}|, \label{eq:haar_integral_simplify}
\end{align}
where $\sigma, \pi$ are permutations in the permutation group $S_n$. The Weingarten function is defined as ${\rm Wg}(\sigma^{-1}\pi,n;d) = (d^{|\sigma^{-1}\pi|})^{-1}$ where $|\sigma^{-1}\pi|$ denotes the number of cycles of the permutation $\sigma^{-1}\pi$. The second line approximation holds in the thermodynamic limit $d \to \infty$ where we only need to focus on  $\sigma = \pi$ in the summation. In this thermodynamic limit, we have ${\rm Wg}(e, n, d) \to 1/d^{n}$.

To close this subsection, we consider the inner product between arbitrary permutation $\sigma \in S_{2m}$ in permutation group $S_{2m}$ and another special designed operator $Q_{2m} \equiv \sum_{z_1, z_2} \ketbra{z_1^{\otimes n}, z_2^{\otimes K}, z_1^{\otimes K}, z_2^{\otimes n}}{z_1^{\otimes n}, z_1^{\otimes K}, z_2^{\otimes K}, z_2^{\otimes n}}_B$. We will see this operator appearing in various places through calculation of frame potentials. The inner product of interest is then
\begin{align}
    (\hat{\sigma}|Q_{2m})_B &= \sum_{z_1, z_2} \braket{z_1^{\otimes n}, z_1^{\otimes K}, z_2^{\otimes K}, z_2^{\otimes m}|\hat{\sigma}^\dagger|z_1^{\otimes n}, z_2^{\otimes K}, z_1^{\otimes K}, z_2^{\otimes n}}\\
    &=\sum_{z_1} \braket{z_1^{\otimes 2n}|\hat{\sigma}|z_1^{\otimes 2n}}+ \sum_{z_1 \neq z_2} \braket{z_1^{\otimes n}, z_1^{\otimes K}, z_2^{\otimes K}, z_2^{\otimes n}|\hat{\sigma}^\dagger|z_1^{\otimes n}, z_2^{\otimes K}, z_1^{\otimes K}, z_2^{\otimes n}}\\
    &= d_B + d_B(d_B-1) \delta_{\sigma, \tau_K(\pi_1\otimes \pi_2)},
    \label{eq:id_inner}
\end{align}
where $\tau_K$ is the transposition operation on the central $2K$ replicas to exchange them, $\tau_K = (1),\dots,(n),(n+1,m+1),\dots,(m,m+K),(m+K+1),\dots,(2m)$, and $\pi_1, \pi_2 \in S_m$ are arbitrary permutations onto the first $m$ and last $m$ replicas separately, which is visualized in Fig.~\ref{fig:id_inner}. Specifically, for $m = 1$, we have
\be
    (\hat{\sigma}|Q_2)_B = \begin{cases}
        d_B, &\text{if $\sigma = e$}\\
        d_B^2, &\text{if $\sigma = \tau_1$},
    \end{cases}
    \label{eq:id_inner_2}
\ee
where $e$ and $\tau_1$ are the identity and swap operation on $2$ replicas.

\begin{figure}[t]
    \centering
    \includegraphics[width=0.7\textwidth]{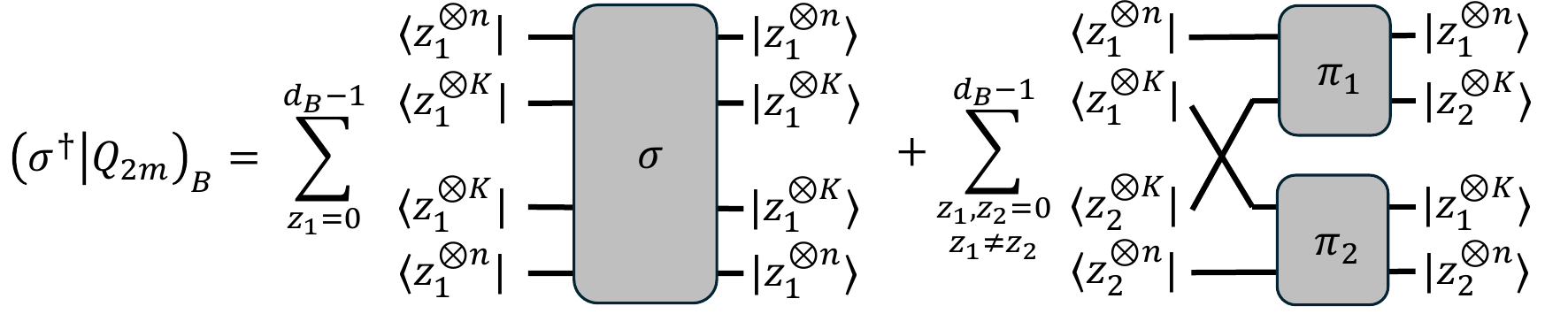}
    \caption{We visualize the identity of $(\hat{\sigma}|Q_{2m})_B$ in Eq.~\eqref{eq:id_inner}.}
    \label{fig:id_inner}
\end{figure}

\subsection{Projected ensemble and frame potential}

To close this section, we reintroduce the definition of projected ensemble and frame potential.

Consider a bipartite quantum state $\ket{\Psi}_{AB}$, we perform projective measurements on one of its subsystems, i.e. $\ketbra{z}{z}_B$, and collect the post-measurement states in subsystem $A$, which forms the so-called {\em projected ensemble} $\calE = \{p_z, \ket{\psi_z}_A\}$ with
\be
    p_z = |{}_B\braket{z|\Psi}_{AB}|^2, \ket{\psi_z}_A = {}_B\braket{z|\Psi}_{AB}/\sqrt{p_z}.
\ee
For convenience of calculation, we define the unnormalized post-measurement state by $\ket{\tilde{\psi}_z}_A \equiv {}_B\braket{z|\Psi}_{AB}$, and the probability for measurement outcome $z$ can be equivalently expressed as $p_z = \braket{\tilde{\psi}_z|\tilde{\psi}_z}$.

The $K$-th frame potential of the projected ensemble $\calE$ is 
\begin{align}
    \calF^{(K)} &= \sum_{z_1,z_2} p_{z_1} p_{z_2} |\braket{\psi_{z'}|\psi_z}|^{2K}\\
    &= \sum_{z_1, z_2} p_{z_1}^{1-K} p_{z_2}^{1-K} |\braket{\tilde{\psi}_{z'}|\tilde{\psi}_z}|^{2K} \\
    &= \sum_{z_1, z_2} \braket{\tilde{\psi}_{z_1}|\tilde{\psi}_{z_1}}^{1-K} |\braket{\tilde{\psi}_{z'}|\tilde{\psi}_z}|^{2K} \braket{\tilde{\psi}_{z_2}|\tilde{\psi}_{z_2}}^{1-K}
    \label{eq:Fk_def}
\end{align}
Specifically, for $K=1$, it reduces to $\calF^{(1)} = \sum_{z_1, z_2} |\braket{\tilde{\psi}_{z'}|\tilde{\psi}_z}|^2$.

\section{State design in deep thermalization}

In this section, we review the results of state design in DT, which can make readers familiar with the methodology for this problem and understand the connection and challenges from DT to HDT.

\subsection{First-order frame potential}

We start by evaluating the first-order frame potential of the projected ensemble from DT following Eq.~\eqref{eq:Fk_def}.
\be
    \calF^{(1)} = \sum_{z_1, z_2} |\braket{\tilde{\psi}_{z_1}|\tilde{\psi}_{z_2}}|^2 = \sum_{z_1, z_2} \tr(\ketbra{\Psi}{\Psi}_{AB}^{\otimes 2} \ketbra{z_1, z_2}{z_2, z_1}_B).
\ee
In DT, the state of the joint system is modeled as a Haar random state $\ket{\Psi}_{AB} = U\ket{\bm 0}_{AB}$ where the unitary follows Haar measure $U \sim \calU_{\rm Haar}(d)$ with $d = d_A d_B$. The ensemble-averaged first-order frame potential becomes
\be
    \E_{\rm Haar}\calF^{(1)} = \E_{\rm Haar} \tr(U^{\otimes 2}\ketbra{\bm 0}{\bm 0}_{AB}^{\otimes 2} U^\dagger{}^{\otimes 2} \sum_{z_1, z_2}\ketbra{z_2, z_1}{z_1, z_2}_B) = {}_A(0|{}_B(0|\E_{\rm Haar}\left[\left(U \otimes U^*\right)^{\otimes 2}\right]|\hat{e})_A|Q_2)_B,
    \label{eq:F1_DT}
\ee
where in the second equation we use Eq.~\eqref{eq:twirling_dual}. Here $\hat{e}$ is the identity operator applied on the $2$-fold Hilbert space, and $|Q_2)_B \equiv \sum_{z_1, z_2} \ket{z_2, z_1}_B \ket{z_1, z_2}_B$. We visualize this equation in a tensor diagram as shown in Fig.~\ref{fig:diagram_DT}a.
According to the Haar unitary identity in Eq.~\eqref{eq:haar_integral}, we can explicitly express the 2-fold Haar unitary twirling as
\be
    \E_{\rm Haar}\left[\left(U \otimes U^*\right)^{\otimes 2}\right] = \frac{1}{d^2 - 1}\left[|\hat{e})(\hat{e}| + |\hat{\tau}_1)(\hat{\tau}_1| - \frac{1}{d}|\hat{e})(\hat{\tau}_1| - \frac{1}{d}|\hat{\tau}_1)(\hat{e}|\right],
\ee
where $\hat{\tau}_1$ is the swap operator. Utilizing this identity, we can derive the ensemble average first-order frame potential
\begin{align}
    &\E_{\rm Haar}\calF_{\rm DT}^{(1)} = {}_A(0|{}_B(0|\E_{\rm Haar}\left[\left(U \otimes U^*\right)^{\otimes 2}\right]|\hat{e})_A|Q_2)_B \nonumber\\
    &= {}_A(0|{}_B(0|\frac{1}{d^2 - 1}\left[|\hat{e})(\hat{e}| + |\hat{\tau}_1)(\hat{\tau}_1| - \frac{1}{d}|\hat{e})(\hat{\tau}_1| - \frac{1}{d}|\hat{\tau}_1)(\hat{e}|\right]|\hat{e})_A|Q_2)_B\\
    &= \frac{1}{d^2 - 1} \left[(\hat{e}|\hat{e})_A(\hat{e}|Q_2)_B (\bm 0|\hat{e})_{AB} + (\hat{\tau}_1|\hat{e})_A(\hat{\tau}_1|Q_2)_B (\bm 0|\hat{\tau}_1)_{AB} - \frac{1}{d}(\hat{\tau}_1|\hat{e})_A (\hat{\tau}_1|Q_2)_B (\bm 0|\hat{e})_{AB} - \frac{1}{d}(\hat{e}|\hat{e})_A(\hat{e}|Q_2)_B (\bm 0|\hat{\tau}_1)_{AB}\right] \\
    &= \frac{1}{d^2 - 1} \left[d_A^2 d_B  + d_A d_B^2 - \frac{1}{d}d_A d_B^2  - \frac{1}{d}d_A^2 d_B\right]\\
    &= \frac{d_B(d_A^2 - 1)}{d_A^2 d_B^2-1} + \frac{d_A^2(d_B^2 - 1)}{d_A^2 d_B^2 - 1}\calF^{(1)}_{\rm Haar}, \label{eq:F1_sol_DT}
\end{align}
where in the second to last line we utilize the fact that $(\hat{\sigma}|\hat{\sigma})_A = d_A^2, (\bm 0|\hat{\sigma})_{AB} = 1, \forall \sigma \in S_2$,  $(\hat{\sigma}|\hat{e})_A = d_A$, and the inner product identity in Eq.~\eqref{eq:id_inner_2}.

\begin{figure}[t]
    \centering
    \includegraphics[width=0.65\textwidth]{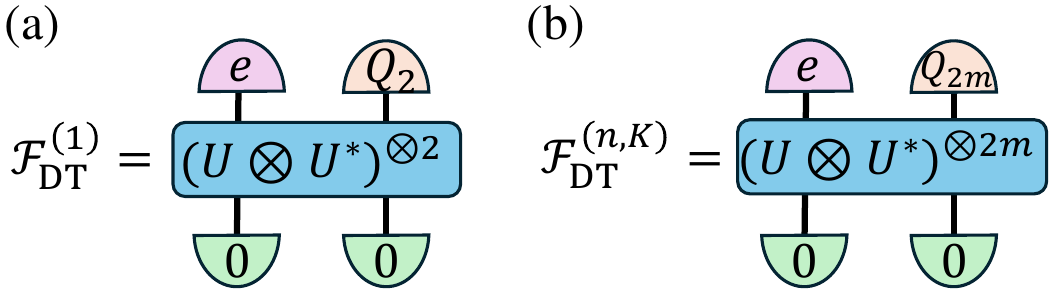}
    \caption{Tensor diagrams for the (a) first-order frame potential of Eq.~\eqref{eq:F1_DT} and (b) pseudo frame potential of Eq.~\eqref{eq:Fk_pseudo_DT} in DT. }
    \label{fig:diagram_DT}
\end{figure}

\subsection{Higher-order frame potential}

In this subsection, we evaluate the ensemble average $K$-th frame potential in projected ensemble of DT with a general assumption that $K \ge 1$. We also recommend interested readers to read Appendix D in Ref.~\cite{ippoliti2022solvable}.

To evaluate the frame potential in Eq.~\eqref{eq:Fk_def}, we utilize the replica trick and define another quantity, pseudo frame potential, 
\begin{align}
    \calF^{(n,K)} 
    &\equiv \sum_{z_1, z_2} p_{z_1}^{n}  |\braket{\tilde{\psi}_{z_1}|\tilde{\psi}_{z_2}}|^{2K} p_{z_2}^{n} = 
    \sum_{z_1, z_2} \braket{\tilde{\psi}_{z_1}|\tilde{\psi}_{z_1}}^{n} |\braket{\tilde{\psi}_{z_1}|\tilde{\psi}_{z_2}}|^{2K} \braket{\tilde{\psi}_{z_2}|\tilde{\psi}_{z_2}}^{n}\\
    &= \tr(\ketbra{\Psi}{\Psi}_{AB}^{\otimes 2m} \sum_{z_1, z_2} \ketbra{z_1^{\otimes n}, z_2^{\otimes K}, z_1^{\otimes K}, z_2^{\otimes n}}{z_1^{\otimes n}, z_1^{\otimes K}, z_2^{\otimes K}, z_2^{\otimes n}}_B).
    \label{eq:Fk_pseudo_def}
\end{align}
where we define $m\equiv n + K$ for simplicity. In the replica trick, we first evaluate the pseudo frame potential and later take the limit $n \to 1-K$ (or $m \to 1$) to obtain the true frame potential.

The ensemble average pseudo frame potential of projected ensemble in DT can be written as
\begin{align}
    \E_{\rm Haar}\calF_{\rm DT}^{(n,K)} &= \E_{\rm Haar}\sum_{z_1, z_2} \tr(U^{\otimes 2m}\ketbra{\bm 0}{\bm 0}_{AB}^{\otimes 2m}U^\dagger{}^{\otimes 2m} \ketbra{z_1^{\otimes n}, z_2^{\otimes K}, z_1^{\otimes K}, z_2^{\otimes n}}{z_1^{\otimes n}, z_1^{\otimes K}, z_2^{\otimes K}, z_2^{\otimes n}}_B)\\
    &= {}_{AB}(\bm 0|\E_{\rm Haar}\left[\left(U \otimes U^*\right)^{\otimes 2m}\right]|\hat{e})_A|Q_{2m})_B,
    \label{eq:Fk_pseudo_DT}
\end{align}
where we define $|Q_{2m})_B \equiv \sum_{z_1, z_2} \ket{z_1^{\otimes n}, z_2^{\otimes K}, z_1^{\otimes K}, z_2^{\otimes n}}_B \ket{z_1^{\otimes n}, z_1^{\otimes K}, z_2^{\otimes K}, z_2^{\otimes n}}_B$ for simplicity and here $\hat{e}$ is the identity operation on $2m$ replicas. We represent this equation in a tensor diagram as shown in Fig.~\ref{fig:diagram_DT}b.

Now we evaluate the ensemble average pseudo frame potential utilizing the approximated Haar unitary twirling identity in Eq.~\eqref{eq:haar_integral_simplify} as
\begin{align}
    \E_{\rm Haar}\calF_{\rm DT}^{(n,K)} &= {}_{AB}(\bm 0|\E_{\rm Haar}\left[\left(U \otimes U^*\right)^{\otimes 2m}\right]|\hat{e})_A|Q_{2m})_B \nonumber\\
    &\simeq {\rm Wg}(e, 2m;d) \sum_{\sigma \in S_{2m}}(\bm 0|\hat{\sigma}){}_{AB}(\hat{\sigma}|\hat{e})_A (\hat{\sigma}|Q_{2m})_B\\
    &={\rm Wg}(e, 2m;d) \left[d_B\sum_{\sigma\in S_{2m}}\tr_A(\hat{\sigma})+d_B (d_B - 1)\sum_{\pi_1, \pi_2 \in S_m} \tr_A(\hat{\tau}_K(\hat{\pi}_1\otimes \hat{\pi}_2))\right].
    \label{eq:Fk_pseudo_DT_1}
\end{align}
The approximation in the second line holds in the thermodynamic limit $d_A d_B \to \infty$, and is widely adopted in related works~\cite{ippoliti2022solvable, chan2024projected}.

Recall the definition of $K$-th moment of Haar state ensemble in Eq.~\eqref{eq:K_moment_haar}, we can reduce the first summation above to
\be
    \sum_{\sigma\in S_{2m}}\tr_A(\hat{\sigma}) = \frac{(2m+d_A - 1)!}{(d_A -1)!} \tr(\rho_{\rm Haar}^{(2m)}) = \frac{(2m+d_A - 1)!}{(d_A -1)!},
    \label{eq:S_2m_trace}
\ee
since $\rho_{\rm Haar}^{(2m)}$ is a valid quantum state. Similarly, for the second summation, we have
\begin{align}
    \sum_{\pi_1, \pi_2 \in S_m}\tr_A(\hat{\tau}_K(\hat{\pi}_1\otimes \hat{\pi}_2)) &= \left(\frac{(m+d_A - 1)!}{(d_A -1)!}\right)^2 \tr_A \left(\hat{\tau}_K \rho_{\rm Haar}^{(m)}{}^{\otimes 2}\right)\\
    &= \left(\frac{(m+d_A - 1)!}{(d_A -1)!}\right)^2 \tr_A \left(\hat{\tau}_K \rho_{\rm Haar}^{(K)}{}^{\otimes 2}\right)= \left(\frac{(m+d_A - 1)!}{(d_A -1)!}\right)^2 \calF^{(K)}_{\rm Haar},
    \label{eq:S_m2_trace}
\end{align}
where from the first equation to the second one we use
\begin{align}
    \tr_A(\hat{\tau}_K (\hat{\pi}_1\otimes \hat{\pi}_2)) &= \int_{\rm Haar} d\psi d\phi \tr_A\left(\hat{\tau}_K \ketbra{\psi^{\otimes n}, \psi^{\otimes K}, \phi^{\otimes K}, \phi^{\otimes n}}{\psi^{\otimes n}, \psi^{\otimes K}, \phi^{\otimes K}, \phi^{\otimes n}} \right) \\
    &= \int_{\rm Haar} d\psi d\phi \tr_A\left(\hat{\tau}_K \ketbra{\psi^{\otimes K}, \phi^{\otimes K}}{\psi^{\otimes K}, \phi^{\otimes K}} \right) = \tr_A\left(\hat{\tau}_K \rho_{\rm Haar}^{(K)}{}^{\otimes 2}\right).
\end{align}
Combining Eqs.~\eqref{eq:S_2m_trace} and~\eqref{eq:S_m2_trace}, we obtain the 
ensemble average pseudo frame potential from Eq.~\eqref{eq:Fk_pseudo_DT_1} as
\begin{align}
    \E_{\rm Haar}\calF_{\rm DT}^{(n,K)} = {\rm Wg}(e, 2m;d) \left[d_B \frac{(2m+d_A - 1)!}{(d_A -1)!}+d_B (d_B - 1) \left(\frac{(m+d_A - 1)!}{(d_A -1)!}\right)^2 \calF^{(K)}_{\rm Haar}\right].
\end{align}
In the last step, we take the limit $m \to 1$ in the replica trick and obtain the ensemble average $K$-th frame potential 
\begin{align}
    \E_{\rm Haar}\calF_{\rm DT}^{(K)} &= {\rm Wg}(e, 2;d) \left[d_B d_A (d_A+1) + d_B (d_B - 1) d_A^2 \calF^{(K)}_{\rm Haar}\right]\\
    &= \frac{d_A+1}{d_A d_B} + \frac{d_B-1}{d_B}\calF^{(K)}_{\rm Haar}, \label{eq:Fk_sol_DT}
\end{align}
which reaches the same conclusion as in Ref.~\cite{ippoliti2022solvable}. We want to point out that in the case $K=1$, Eq.~\eqref{eq:Fk_sol_DT} becomes $(d_A+1)/d_Ad_B+d_A(d_B-1)\calF^{(1)}_{\rm Haar}/d_A d_B$, which is different from the exact result in Eq.~\eqref{eq:F1_sol_DT} up to a small correction of $(d_A+d_B)/d_Ad_B(d_A d_B+1)\sim 1/d_A d_B^2 + 1/d_A^2 d_B$ due to the approximation in Eq.~\eqref{eq:haar_integral_simplify}.  

The necessary number of ancilla qubits to achieve an $\epsilon$-approximate $K$-design ($\calF_{\rm DT}^{(K)} \le (1+\epsilon)\calF_{\rm Haar}^{(K)}$) is
\begin{align}
    N_B &\ge 
    \log_2\left(\frac{d_A+1 -d_A \calF_{\rm Haar}^{(K)} }{d_A \calF_{\rm Haar}^{(K)} \epsilon }\right) \\
    &=\log_2\left(\frac{1}{\calF_{\rm Haar}^{(K)}}\right) + \log_2\left(\frac{1}{\epsilon}\right) + \log_2\left(\left(1 +\frac{1}{d_A}\right)\left(1- \frac{d_A \calF_{\rm Haar}^{(K)}}{d_A+1}\right)\right)\\
    &= KN_A - \log_2(K!) + \log_2\left(1/\epsilon\right) + 
    \calO(2^{-N_A}),
\end{align}
where in the last line we take the thermodynamic limit of $d_A = 2^{N_A} \to \infty$ thus $\calF_{\rm Haar}^{(K)} = K!/d_A^K$.

\section{State design in holographic deep thermalization}
\label{app:F_dynamics}

In this section, we study the frame potential of the projected ensemble produced by HDT at an arbitrary time step $t$. Similar to the DT case, we model each unitary $U_t$ in HDT to be randomly sampled from the Haar measure $\calU_\text{Haar}(d)$ with $d = d_A d_B$.
For HDT, we postpone the projective measurements on the bath at different steps to the end by equivalently expanding the bath system $\bm B = B_1 \cdots B_t$ (see Fig.~\ref{fig:diagram_HDT}a). Therefore, the output state can be expressed as
\be
    \ket{\Psi}_{A \bm B} = \prod_{i=1}^{t} U_{AB_i}\ket{\bm 0}_{A\bm B} \equiv U\ket{\bm 0}_{A \bm B},
    \label{eq:HDT_state}
\ee
where $U_{AB_i}$ is a unitary {\em only} nontrivially applied on the data system $A$ and bath system $B_i$. We also collect all the measurement outcome into a vector representation $\bfz = (z_1, \cdots, z_t)^T$ for simplicity, and the unnormalized post-measurement state becomes $\ket{\tilde{\psi}_\bfz} = {}_{\bm B}\braket{\bfz|\Psi}_{A \bm B}$.

\begin{figure}
    \centering
    \includegraphics[width=0.65\linewidth]{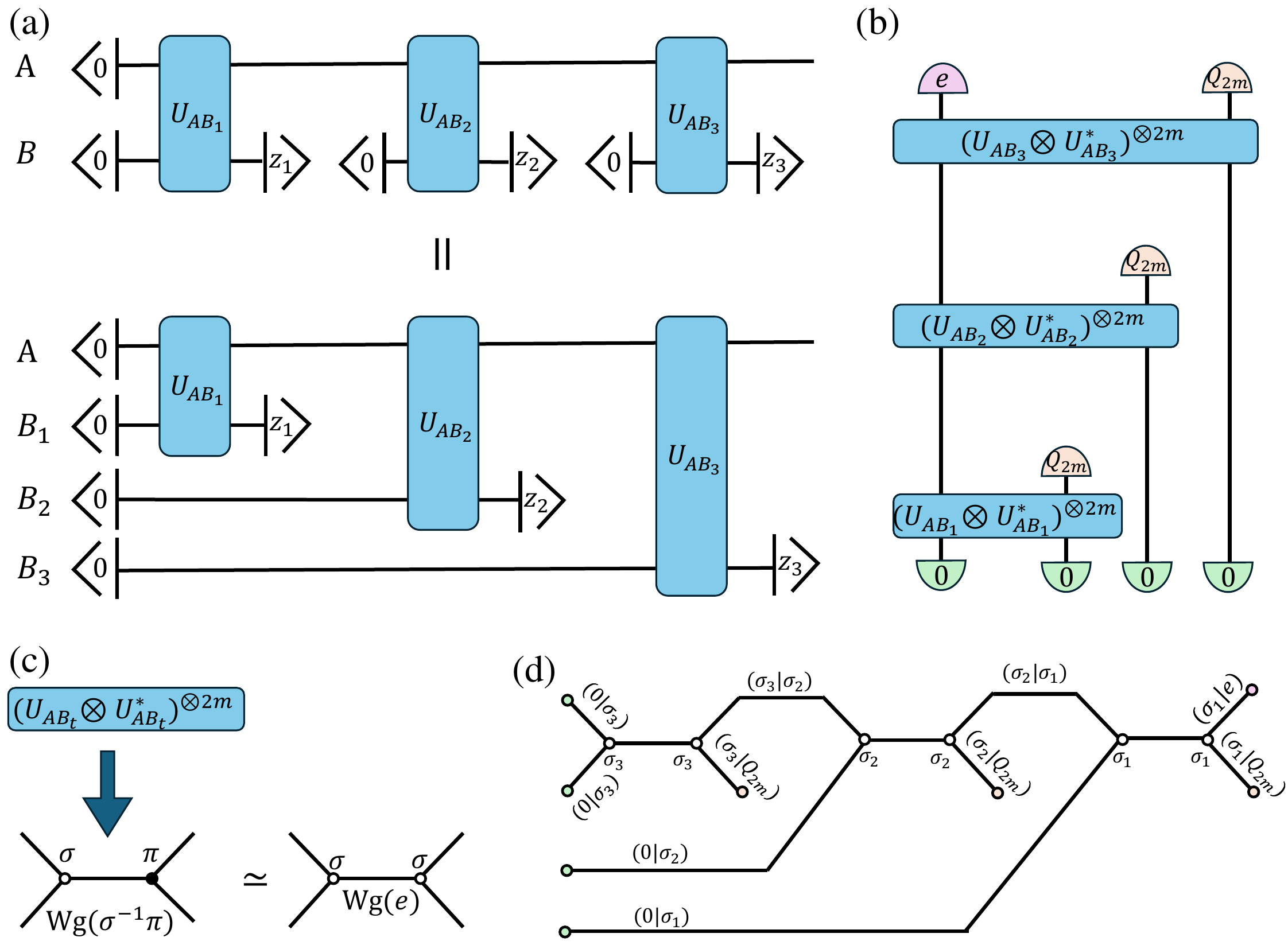}
    \caption{Diagrams for HDT. In (a), we postpone all mid-circuit measurements in HDT to the end and equivalently expand the bath system. In (b), we show the tensor diagram representation of pseudo frame potential in HDT. In (c), we map the unitary twirling to a ``spin'' representation in statistical model. The approximation forces the two ``spins'' to be the same. In (d), we express the pseudo frame potential in (b) as the partition function of a statistical model with weights determined by unitary twirling and permutation inner products (see notations). The time direction of HDT is from left to right, while our evaluation of pseudo frame potential starts follows the opposite direction. In all plots, we show a case of $t = 3$.}
    \label{fig:diagram_HDT}
\end{figure}
\subsection{First-order frame potential}

We begin with the first-order frame potential in HDT. Following the result in DT (see Eq.~\eqref{eq:F1_DT}), the ensemble average first-order frame potential becomes
\be
    \E_{\rm Haar}\calF^{(1)}_{\rm HDT}(t) = {}_A(0|\left(\otimes_{i=1}^t {}_{B_i}(0|\right) \prod_{i=1}^{t}\E_{\rm Haar}\left[\left(U_{AB_i} \otimes U_{AB_i}^*\right)^{\otimes 2}\right]|\hat{e})_A \otimes_{i=1}^t |Q_2)_{B_i}.
    \label{eq:F1_HDT}
\ee
We visualize it as a tensor diagram in  Fig.~\ref{fig:diagram_HDT}b with $m = 1$. 
We first evaluate the unitary twirling of $U_{AB_t}$ on the boundary condition as
\be
    {}_{B_t}(0|\E_{\rm Haar}[\left(U_{AB_t} \otimes U_{AB_t}^*\right)^{\otimes 2}]|\hat{e})_A |Q_2)_{B_t} = \frac{d_B(d_A^2 - 1)}{d_A^2 d_B^2-1}|\hat{e})_A + \frac{d_A(d_B^2 - 1)}{d_A^2 d_B^2 - 1}|\hat{\tau}_1)_A,
    \label{eq:F1_rule1_HDT}
\ee
which is identical to the result in DT (see Eq.~\eqref{eq:F1_sol_DT}).
Since a new boundary condition $|\hat{\tau}_1)_A$ appears, we next evaluate the unitary twirling of $U_{AB_{t-1}}$ on it as
\begin{align}
    &{}_{B_{t-1}}(0|\E_{\rm Haar}\left[\left(U_{AB_{t-1}} \otimes U_{AB_{t-1}}^*\right)^{\otimes 2}\right]|\hat{\tau}_1)_A|Q_2)_{B_{t-1}} = \frac{1}{d^2 - 1}{}_{B_t}( 0|\left[|\hat{e})(\hat{e}| + |\hat{\tau}_1)(\hat{\tau}_1| - \frac{1}{d}|\hat{e})(\hat{\tau}_1| - \frac{1}{d}|\hat{\tau}_1)(\hat{e}|\right]|\hat{\tau}_1)_A|Q_2)_{B_{t-1}}\\
    &= \frac{1}{d^2 - 1} \left[d_A d_B|\hat{e})_A + d_A^2 d_B^2|\hat{\tau}_1)_A - \frac{d_A^2 d_B^2}{d}|\hat{e})_A - \frac{d_A d_B}{d}|\hat{\tau}_1)_A\right]\\
    &= |\hat{\tau}_1)_A,
    \label{eq:F1_rule2_HDT}
\end{align}
which is invariant under the unitary twirling. 

At this point, we have collected all the necessary results to understand the unitary twirling. For calculation convenience, we can represent the above two boundary conditions using a $2$-dimensional linear system, $\ket{\underline{0}} := | \hat{e})_A$ and $\ket{\underline{1}} := |\hat{\tau}_1)_A$, the twirling results in Eqs. \ref{eq:F1_rule1_HDT} and \eqref{eq:F1_rule2_HDT} can be represented by a $2\times 2$ matrix
\begin{align}
    \calA = \frac{1}{d_A^2 d_B^2 -1}\begin{pmatrix}
        d_B(d_A^2-1) &0 \\
        d_A(d_B^2-1) & d_A^2 d_B^2-1
    \end{pmatrix}.
\end{align}
With the matrix representation, the ensemble-averaged frame potential can be derived as
\begin{align}
    \E_{\rm Haar}\calF^{(1)}_{\rm HDT}(t) &= \braket{\underline{0}|\calA^t|\underline{0}}(0|\hat{e})_A + \braket{\underline{1}|\calA^t|\underline{0}}(0|\hat{\tau}_1)_A\\
    &= \frac{(d_A-1) (d_A d_B-1)}{d_A^2 d_B + 1} \left(\frac{\left(d_A^2-1\right) d_B}{d_A^2 d_B^2-1}\right)^t+ \frac{d_A^2 (d_B+1)}{d_A^2 d_B+1} \calF^{(1)}_{\rm Haar}.
    \label{eq:F1_sol_HDT}
\end{align}
Specifically, by taking $t=1$, we can recover the DT result in Eq.~\eqref{eq:F1_sol_DT}.

\subsection{Higher-order frame potential}

In this subsection, we analyze the dynamics of the higher-order frame potential in HDT. Following the result in DT (see Eq.~\eqref{eq:Fk_pseudo_DT}), the ensemble average of pseudo frame potential is
\be
    \E_{\rm Haar}\calF^{(n,K)}_{\rm HDT}(t) = {}_A(0|\left(\otimes_{i=1}^t {}_{B_i}(0|\right) \prod_{i=1}^{t}\E_{\rm Haar}\left[\left(U_{AB_i} \otimes U_{AB_i}^*\right)^{\otimes 2m}\right]|\hat{e})_A \otimes_{i=1}^t |Q_{2m})_{B_i}.
    \label{eq:Fk_pseudo_HDT}
\ee
We visualize it as a tensor diagram in  Fig.~\ref{fig:diagram_HDT}b. Furthermore, we can map the Haar unitary twirling in Eq.~\eqref{eq:haar_integral} to a ``spin'' representation with black and white vertices denoting the permutations $\sigma$ and $\pi$~\cite{zhou2019emergent, hunter2019unitary}. The horizontal line connecting the vertices represents the Weingarten function ${\rm Wg}(\sigma, 2m;d)$. In the large system size limit $d_A d_B \to \infty$, we can do a further simplification with $\sigma = \pi$ (white vertices only) and the horizontal line now becomes ${\rm Wg}(e, 2m;d)$. With this statistical mapping, we can map the pseudo frame potential in Eq.~\eqref{eq:Fk_pseudo_HDT} into the partition function of a statistical physics model, shown in Fig.~\ref{fig:diagram_HDT}d. The white vertex represents a permutation $\sigma \in S_{2m}$ as stated above. The colored vertices represent the corresponding boundary condition shown in Fig.~\ref{fig:diagram_HDT}b. The direct horizontal line connecting two white vertices represents ${\rm Wg}(\sigma, 2m;d)$, and the inclined line or polyline are inner products between permutations or boundaries, as marked in the Fig.~\ref{fig:diagram_HDT}d. Therefore, we can express the pseudo frame potential in the large system limit ($d_A d_B \to \infty$) as
\be
    \E_{\rm Haar}\calF^{(n,K)}_{\rm HDT}(t) = {\rm Wg}(e,2m;d)^t \sum_{\sigma_1, \dots, \sigma_t \in S_{2m}} \left[(0|\hat{\sigma}_t)_A\prod_{i=1}^{t-1} (\hat{\sigma}_{i+1}|\hat{\sigma}_{i})_A (\hat{\sigma}_1|e)_A \right]\left[\prod_{i=1}^{t} (0|\hat{\sigma}_i)_{B_i} (\hat{\sigma}_{i}|Q_{2m})_{B_i} \right].
    \label{eq:Fk_pseudo_HDT_stat}
\ee
Calculating the above quantity for the desired frame potential is challenging due to the iterative interaction between different permutations $(\hat{\sigma}_{i+1}|\hat{\sigma})_A$, which is absent in the calculation of conventional DT in Eq.~\eqref{eq:Fk_pseudo_DT}. In the following, we first evaluate a summation on a subset of $\sigma_1, \dots, \sigma_t \in S_{2m}$ to obtain a lower bound on the pseudo frame potential to provide analytical insight.
Before starting the calculation, we first introduce a lemma to decompose the permutation group $S_{2m}$ into two partitioning subsets.
\begin{lemma}\label{lem:permutation_decomp}
    For any permutation $\sigma \in S_{2m}$, we can express it as
\be
    \sigma = \omega (\pi_1 \otimes \pi_2),
\ee
where $\pi_1, \pi_2 \in S_m$ are two permutations onto the first and last $m$ replicas, and $\omega$ is a product of disjoint swaps crossing the first and last $m$ replicas as $\omega = \prod_{j_1 \in C_1, j_2 \in C_2} \tau_{(j_1 j_2)}$ with $C_1 \subset B_1 :=\{1, \dots, m\}$ and $C_2 \subset B_2 := \{m+1, \dots, 2m\}$.
\end{lemma}

\begin{proof}[Proof of Lemma~\ref{lem:permutation_decomp}]
    For any permutation $\sigma \in S_{2m}$, we first partition $\sigma(\{1,\dots,2m\}) = (\sigma_1, \sigma_2)$ into two blocks of size $m$. For $i \in \{1,2\}$, let $A_i = \sigma_i \cap B_i$ be the subset of entries staying in the $i$-th permutation block. Denote $C_i = B_i \setminus A_i$. Clearly, there are $|C_1| = m - |A_1|$ entries coming from $B_2$, which implies that $m - |C_1| = |A_1|$ entries staying in $B_2$. So $|A_1| = |A_2| \leq m$ and we can choose a bijection $f : C_1 \to C_2$ such that for each $j \in C_1$, we have $f(j) \in C_2$, and vice versa. Then, $(A_1, f^{-1}(C_2))$ is a permutation $\pi_1$ of $B_1$ and $(A_2, C_2)$ is a permutation $\pi_2$ of $B_2$. Since the bijection $f$ can be represented by a product of disjoint swaps, this proves the lemma.
\end{proof}

One can do a simple cardinality counting for the sanity check on this Lemma. For the disjoint swap $\omega$ with $|C_1| = |C_2| = r$, the number of these swaps are $\binom{m}{r}^2$, and the number of $(\pi_1 \otimes \pi_2)$ is $m!^2$, therefore the number of permutations of $\{\omega(\pi_1 \otimes \pi_2)\}$ is $m!^2 \sum_{r=0}^m \binom{m}{r}^2 = (2m)!$, which is exactly the cardinality of permutation group $S_{2m}$. For convenience of notion, we denote $G_r \equiv \{\omega (\pi_1 \otimes \pi_2)| |C_1|=|C_2|=r\}$. We can further have a corollary which directly follows from Eq.~\eqref{eq:id_inner}.
\begin{corollary}
\label{cor:id_inner_decomposition}
    For any permutation $\sigma \in S_{2m}$, we have
    \be
    (\hat{\sigma}|Q_{2m})_B = \begin{cases}
        d_B^2, 
        &\text{if $\sigma \in G_K'$} \\
        d_B, &\text{otherwise}
    \end{cases},
    \ee
    where $G_K' \equiv \{\tau_K(\pi_1 \otimes \pi_2)|\pi_1, \pi_2 \in S_m\}$.
\end{corollary}

Utilizing the decomposition Lemma~\ref{lem:permutation_decomp} and Corrollary~\ref{cor:id_inner_decomposition}, we can re-express the pseudo frame potential and write out a lower bound for $t \ge 2$ as
\begin{align}
    &\E_{\rm Haar}\calF^{(n,K)}_{\rm HDT}(t) = {\rm Wg}(e,2m;d)^t \sum_{\sigma_1, \dots, \sigma_t \in \left(\cup_{r=0}^m G_r\right) } \left[\prod_{i=1}^{t-1} (\hat{\sigma}_{i+1}|\hat{\sigma}_{i})_A (\hat{\sigma}_1|e)_A \right]\left[\prod_{i=1}^{t} \left(d_B + d_B(d_B-1)\delta_{\sigma_i \in G_m} \right)\right]\\
    & \ge {\rm Wg}(e,2m;d)^t \left[d_B^t  \sum_{\sigma_1, \dots, \sigma_t \in G_0} \prod_{i=1}^{t-1} (\hat{\sigma}_{i+1}|\hat{\sigma}_{i})_A (\hat{\sigma}_1|e)_A + d_B^{2t} \sum_{\sigma_1, \dots, \sigma_t \in G_K'}\prod_{i=1}^{t-1} (\hat{\sigma}_{i+1}|\hat{\sigma}_{i})_A (\hat{\sigma}_1|e)_A \nonumber\right.\\
    &\left. \qquad \qquad \qquad \qquad+ d_B^{2t-1} \sum_{\substack{\sigma_1 \in \left(\cup_{r=0}^{K} G_r\right)\setminus G_K' \\
    \sigma_2, \dots, \sigma_t \in G_K'}} \prod_{i=1}^{t-1} (\hat{\sigma}_{i+1}|\hat{\sigma}_{i})_A (\hat{\sigma}_1|e)_A\right],
    \label{eq:Fk_pseudo_HDT_lb}
\end{align}
where $G_K' \equiv \{\tau_K(\pi_1 \otimes \pi_2)|\pi_1, \pi_2 \in S_m\} \in G_K$ is a subset of $G_K$ to maximize Eq.~\eqref{eq:id_inner} to be $d_B^2$, as stated in the Corollary above.
This lower bound naturally holds as all inner products $(\hat{\sigma}_{i+1}|\hat{\sigma}_{i})_A$ are always non-negative. Before proceeding with further calculations, we provide a physical interpretation of the three terms. For the first term, as the initial boundary of calculation in Fig.~\ref{fig:diagram_HDT}d on the subsystem $A$ is $|e)_A \in G_0$, we can regard it as a dynamical term, which only exists in the HDT scheme. The second term involves the largest contribution of bath system of $d_B^{2t}$ in each step, and according to the previous calculation in DT (see Eq.~\eqref{eq:Fk_pseudo_DT_1}), it corresponds to the leading order contribution for the asymptotic frame potential. The above two terms both correspond to the 0-domain wall picture as the permutations through the temporal trajectory chain remain the same.
Finally, for the last one, we iterate $\sigma_1$ over all $(\cup_{r=0}^{K} G_r)\setminus G_K'$ while fixing the rest permutations $\sigma_2, \cdots \sigma_t$ in the set $G_K'$ only, which can be interpreted as an 1-domain wall (1-DW) picture. 

In the following, we evaluate the three terms in Eq.~\eqref{eq:Fk_pseudo_HDT_lb} separately to derive a lower bound for the pseudo frame potential.

For the first term in Eq.~\eqref{eq:Fk_pseudo_HDT_lb}, we have
\begin{align}
   & d_B^t  \sum_{\sigma_1, \dots, \sigma_t \in G_0} \prod_{i=1}^{t-1} (\hat{\sigma}_{i+1}|\hat{\sigma}_{i})_A (\hat{\sigma}_1|e)_A = d_B^t  \sum_{\sigma_1, \dots, \sigma_t \in G_0}   \prod_{i=1}^{t-1}  \tr_A(\hat{\sigma}_{i+1}^\dagger \hat{\sigma}_{i})  \tr_A(\hat{\sigma}_1^\dagger) \\
   &=  d_B^t  \sum_{\sigma_1, \dots, \sigma_{t-1} \in G_0} \tr_A(\hat{\sigma}_1^\dagger) \prod_{i=1}^{t-2} \tr_A(\hat{\sigma}_{i+1}^\dagger \hat{\sigma}_{i}) \sum_{\hat{\sigma}_t \in G_0} \tr_A(\hat{\sigma}_t^\dagger \hat{\sigma}_{t-1}) \\
   &= d_B^t  \sum_{\sigma_1, \dots, \sigma_{t-1} \in G_0} \tr_A(\hat{\sigma}_1^\dagger) \prod_{i=1}^{t-2} \tr_A(\hat{\sigma}_{i+1}^\dagger \hat{\sigma}_{i}) \sum_{\pi_t, \pi_t' \in S_m} \tr_A\left(\left(\hat{\pi}_t^\dagger  \otimes \hat{\pi}_t'^\dagger\right) \left(\hat{\pi}_{t-1} \otimes \hat{\pi}_{t-1}\right)\right) \\
    &= d_B^t  \sum_{\sigma_1, \dots, \sigma_{t-1} \in G_0} \tr_A(\hat{\sigma}_1^\dagger) \prod_{i=1}^{t-2} \tr_A(\hat{\sigma}_{i+1}^\dagger \hat{\sigma}_{i}) \sum_{\tilde{\pi}_t, \tilde{\pi}_t' \in S_m} \tr_A\left(\hat{\tilde{\pi}}_t\otimes \hat{\tilde{\pi}}_t'\right) \\
    &= d_B^t  \sum_{\sigma_1, \dots, \sigma_{t-1} \in G_0} \tr_A(\hat{\sigma}_1^\dagger) \prod_{i=1}^{t-2} \tr_A(\hat{\sigma}_{i+1}^\dagger \hat{\sigma}_{i}) \left(\frac{(m+d_A-1)!}{(d_A-1)!}\right)^2 \tr_A\left(\rho_{\rm Haar}^{(m)}\right)^2 \\
    &\quad \vdots \nonumber \\
    &= d_B^t   \left[\left(\frac{(m+d_A-1)!}{(d_A-1)!}\right)^2 \tr_A\left(\rho_{\rm Haar}^{(m)}\right)^2\right]^t\\
    &= d_B^t   \left(\frac{(m+d_A-1)!}{(d_A-1)!}\right)^{2t},
    \label{eq:Fk_pseudo_HDT_lb_1}
\end{align}
where in the third line we apply the definition of $G_0$ and in the forth line we utilize the property that permutation group is closed. In the fifth line, we utilize the definition of Haar state moment in Eq.~\eqref{eq:K_moment_haar}. The vertical dots represent repeating the above process over the rest permutations from $\sigma_{t-1}$ to $\sigma_1$. 

For the second term in Eq.~\eqref{eq:Fk_pseudo_HDT_lb}, we have
\begin{align}
    &d_B^{2t} \sum_{\sigma_1, \dots, \sigma_t \in G_K'}\prod_{i=1}^{t-1} (\hat{\sigma}_{i+1}|\hat{\sigma}_{i})_A (\hat{\sigma}_1|e)_A = d_B^{2t} \sum_{\sigma_1, \dots, \sigma_t \in G_K'}\prod_{i=1}^{t-1} \tr_A(\hat{\sigma}_{i+1}^\dagger \hat{\sigma}_{i}) \tr_A(\hat{\sigma}_1^\dagger) \\
    &= d_B^{2t} \sum_{\sigma_1, \dots, \sigma_{t-1} \in G_K'}\tr_A(\hat{\sigma}_1^\dagger) \prod_{i=1}^{t-2} \tr_A(\hat{\sigma}_{i+1}^\dagger \hat{\sigma}_{i}) \sum_{\sigma_t \in G_K'} \tr_A(\hat{\sigma}_t^\dagger \hat{\sigma}_{t-1})\\
    &= d_B^{2t} \sum_{\sigma_1, \dots, \sigma_{t-1} \in G_K'}\tr_A(\hat{\sigma}_1^\dagger) \prod_{i=1}^{t-2} \tr_A(\hat{\sigma}_{i+1}^\dagger \hat{\sigma}_{i}) \sum_{\pi_t, \pi'_t \in S_m} \tr_A\left(\left(\hat{\pi}_t^\dagger \otimes \hat{\pi}_t'^\dagger\right) \hat{\tau}_K^\dagger \hat{\tau}_K \left(\hat{\pi}_{t-1} \otimes \hat{\pi}'_{t-1}\right)\right)\\
    &= d_B^{2t} \sum_{\sigma_1, \dots, \sigma_{t-1} \in G_K'}\tr_A(\hat{\sigma}_1^\dagger) \prod_{i=1}^{t-2} \tr_A(\hat{\sigma}_{i+1}^\dagger \hat{\sigma}_{i}) \sum_{\tilde{\pi}_t, \tilde{\pi}'_t \in S_m} \tr_A\left(\hat{\tilde{\pi}}_{t} \otimes \hat{\tilde{\pi}}'_{t}\right)\\
    &= d_B^{2t} \sum_{\sigma_1, \dots, \sigma_{t-1} \in G_K'}\tr_A(\hat{\sigma}_1^\dagger) \prod_{i=1}^{t-2} \tr_A(\hat{\sigma}_{i+1}^\dagger \hat{\sigma}_{i}) \left(\frac{(m+d_A-1)!}{(d_A-1)!}\right)^2 \tr_A\left(\rho_{\rm Haar}^{(m)}\right)^2 \\
    &\quad \vdots \nonumber\\
    &= d_B^{2t} \sum_{\sigma_1 \in  G_K'}\tr_A(\hat{\sigma}_1^\dagger) \left(\frac{(m+d_A-1)!}{(d_A-1)!}\right)^{2(t-1)} \tr_A\left(\rho_{\rm Haar}^{(m)}\right)^{2(t-1)} \\
    &= d_B^{2t}  \left(\frac{(m+d_A-1)!}{(d_A-1)!}\right)^{2} \calF_{\rm Haar}^{(K)} \left(\frac{(m+d_A-1)!}{(d_A-1)!}\right)^{2(t-1)} \tr_A\left(\rho_{\rm Haar}^{(m)}\right)^{2(t-1)} \\
    &= d_B^{2t} \left(\frac{(m+d_A-1)!}{(d_A-1)!}\right)^{2t} \calF_{\rm Haar}^{(K)},
    \label{eq:Fk_pseudo_HDT_lb_2}
\end{align}
where in the third line we apply the definition of $G_K'$ and in the forth line we utilize the property that permutation group is closed. In the fifth line, we utilize the definition of Haar state moment in Eq.~\eqref{eq:K_moment_haar}. The vertical dots represent repeating the above process over the rest permutations from $\sigma_{t-1}$ to $\sigma_2$. In the second to last line, we utilize the result already derived in Eq.~\eqref{eq:S_m2_trace}.

Finally, for the last term (1-DW approximation) in Eq.~\eqref{eq:Fk_pseudo_HDT_lb}, we have
\begin{align}
    &d_B^{2t-1} \sum_{\substack{\sigma_1 \in \left(\cup_{r=0}^{K} G_r\right)\setminus G_K' \\
    \sigma_2, \dots, \sigma_t \in G_K'}} \prod_{i=1}^{t-1} (\hat{\sigma}_{i+1}|\hat{\sigma}_{i})_A (\hat{\sigma}_1|e)_A = d_B^{2t-1} \sum_{\substack{\sigma_1 \in \left(\cup_{r=0}^{K} G_r\right)\setminus G_K' \\
    \sigma_2, \dots, \sigma_t \in G_K'}} \prod_{i=1}^{t-1} \tr_A(\hat{\sigma}_{i+1}^\dagger\hat{\sigma}_{i}) \tr_A(\hat{\sigma}_1^\dagger)\\
    &= d_B^{2t-1} \sum_{\substack{\sigma_1 \in \left(\cup_{r=0}^{K} G_r\right)\setminus G_K' \\
    \sigma_2 \in G_K'}}
     \tr_A(\hat{\sigma}_1^\dagger)  \tr_A(\hat{\sigma}_2^\dagger \hat{\sigma}_1) \sum_{\sigma_3, \dots, \sigma_{t} \in G_K'}\prod_{i=2}^{t-1} \tr_A(\hat{\sigma}_{i+1}^\dagger\hat{\sigma}_{i}) \\
    &= d_B^{2t-1}  \left(\frac{(m+d_A-1)!}{(d_A-1)!}\right)^{2(t-2)} \tr_A\left(\rho_{\rm Haar}^{(m)}\right)^{2(t-2)} \sum_{\substack{\sigma_1 \in \left(\cup_{r=0}^{K} G_r\right) \setminus G_K' \\
    \sigma_2 \in G_K'}} \tr_A(\hat{\sigma}_1^\dagger) \tr_A(\hat{\sigma}_2^\dagger \hat{\sigma}_1) \\
    &= d_B^{2t-1}  \left(\frac{(m+d_A-1)!}{(d_A-1)!}\right)^{2(t-2)} \left[ \left(\frac{(m+d_A-1)!}{(d_A-1)!}\right)^{4} \calF_{\rm Haar}^{(K)} +  \sum_{\substack{\sigma_1 \in \left(\cup_{r=1}^{K} G_r\right) \setminus G_K' \\
    \sigma_2 \in G_K'}} \tr_A(\hat{\sigma}_1^\dagger)\tr_A(\hat{\sigma}_2^\dagger \hat{\sigma}_1)\right],
    \label{eq:Fk_pseudo_HDT_lb_3}
\end{align}
where in the second line we utilize the result through derivation of Eq.~\eqref{eq:Fk_pseudo_HDT_lb_2}. The first term in the last line corresponds to $\sigma_1 \in G_0, \sigma_2 \in G_K'$, which 
can be solved by utilizing the closed group property again and Eq.~\eqref{eq:S_m2_trace}.
Now we are left with a summation, and we first calculate $\sigma_1 \in G_1, \sigma_2 \in G_K'$ to provide insight into the full summation. 
For convenience, we divide the total $2m = 2(n+K)$ replicas into four sets as $D_1, D_2, D_3, D_4$ to denote indices in the corresponding range of $[1, n], [n+1, m], [m+1, m+K], [m+K+1,2m]$. 

For $\sigma_1 \in G_1$, we have
\begin{align}
    &\sum_{\substack{\sigma_1 \in G_1 \\
    \sigma_2 \in G_K'}} \tr_A(\hat{\sigma}_1^\dagger)\tr_A(\hat{\sigma}_2^\dagger \hat{\sigma}_1) = \sum_{\omega}\sum_{\pi_1, \pi_1', \pi_2, \pi_2'}  \tr_A\left(\left(\hat{\pi}_1^\dagger \otimes \hat{\pi}_1'^\dagger\right)\hat{\omega}^\dagger\right) \tr_A\left(\left(\hat{\pi}_2^\dagger \otimes \hat{\pi}_2'^\dagger\right) \hat{\tau}_K^\dagger \hat{\omega} \left(\hat{\pi}_1 \otimes \hat{\pi}_1'\right)\right) \\
    &= \left(\frac{(m+d_A-1)!}{(d_A-1)!}\right)^{2} \sum_{\omega} \sum_{\pi_1, \pi_1' \in S_m} \tr_A\left(\left(\hat{\pi}_1^\dagger \otimes \hat{\pi}_1'^\dagger\right)\hat{\omega}^\dagger\right) \tr_A\left(\hat{\omega} \rho_{\rm Haar}^{(m)}{}^{\otimes 2} \hat{\tau}^\dagger_K\right)\\
    &= \left(\frac{(m+d_A-1)!}{(d_A-1)!}\right)^{4} \left(\sum_{D_1\times D_3} + \sum_{D_2 \times D_4} + \sum_{D_1 \times D_4} + \sum_{D_2 \times D_3}\right) \tr_A\left(\rho_{\rm Haar}^{(m)}{}^{\otimes 2}\hat{\omega}^\dagger\right) \tr_A\left(\hat{\omega} \rho_{\rm Haar}^{(m)}{}^{\otimes 2} \hat{\tau}^\dagger_K\right),
    \label{eq:perm_sum_G1}
\end{align}
where in the first line we apply the definition of $G_1$ and in the last line the summation $D_i \times D_j$ represent the case where the swap in $\omega$ exchanges replicas from $D_i$ and $D_j$. For $D_1 \times D_3$, it becomes
\begin{align}
    &\sum_{D_1 \times D_3} \tr_A\left( \rho_{\rm Haar}^{(m)}{}^{\otimes 2}\hat{\omega}^\dagger\right) \tr_A\left(\hat{\omega} \rho_{\rm Haar}^{(m)}{}^{\otimes 2} \hat{\tau}^\dagger_K\right) \nonumber\\
    &= \sum_{D_1 \times D_3} \int_{\rm Haar} d\psi d\phi \tr_A\left( \ketbra{\psi^{\otimes m}, \phi^{\otimes m}}{\psi^{\otimes m}, \phi^{\otimes m}} \hat{\omega}^\dagger\right)
    \int_{\rm Haar} d\psi d\phi \tr_A\left(\hat{\omega} \ketbra{\psi^{\otimes m}, \phi^{\otimes m}}{\psi^{\otimes m}, \phi^{\otimes m}}\hat{\tau}_K^\dagger \right)\\
    &=\sum_{D_1 \times D_3} \calF_{\rm Haar}^{(1)} \calF_{\rm Haar}^{(K)} = \frac{n K}{d_A} \calF_{\rm Haar}^{(K)}.
\end{align}
Similarly, one can check that 
\be
    \sum_{D_2 \times D_4} \tr_A\left( \rho_{\rm Haar}^{(m)}{}^{\otimes 2}\hat{\omega}^\dagger\right) \tr_A\left(\hat{\omega} \rho_{\rm Haar}^{(m)}{}^{\otimes 2} \hat{\tau}^\dagger_K\right) = \frac{n K}{d_A} \calF_{\rm Haar}^{(K)}.
\ee
Next, for $D_1 \times D_4$, we have
\begin{align}
    &\sum_{D_1 \times D_4} \tr_A\left( \rho_{\rm Haar}^{(m)}{}^{\otimes 2}\hat{\omega}^\dagger\right) \tr_A\left(\hat{\omega} \rho_{\rm Haar}^{(m)}{}^{\otimes 2} \hat{\tau}^\dagger_K\right) \nonumber\\
    &= \sum_{D_1 \times D_4} \int_{\rm Haar} d\psi d\phi \tr_A\left( \ketbra{\psi^{\otimes m}, \phi^{\otimes m}}{\psi^{\otimes m}, \phi^{\otimes m}} \hat{\omega}^\dagger\right)
    \int_{\rm Haar} d\psi d\phi \tr_A\left(\hat{\omega} \ketbra{\psi^{\otimes m}, \phi^{\otimes m}}{\psi^{\otimes m}, \phi^{\otimes m}}\hat{\tau}_K^\dagger \right)\\
    &=\sum_{D_1 \times D_4} \calF_{\rm Haar}^{(1)} \calF_{\rm Haar}^{(K+1)} = \frac{n^2 (K+1)}{d_A (d_A + K)} \calF_{\rm Haar}^{(K)},
\end{align}
and the last one with $D_2 \times D_3$, 
\begin{align}
    &\sum_{D_2 \times D_3} \tr_A\left( \rho_{\rm Haar}^{(m)}{}^{\otimes 2}\hat{\omega}^\dagger\right) \tr_A\left(\hat{\omega} \rho_{\rm Haar}^{(m)}{}^{\otimes 2} \hat{\tau}^\dagger_K\right) \nonumber\\
    &= \sum_{D_2 \times D_3} \int_{\rm Haar} d\psi d\phi \tr_A\left( \ketbra{\psi^{\otimes m}, \phi^{\otimes m}}{\psi^{\otimes m}, \phi^{\otimes m}} \hat{\omega}^\dagger\right)
    \int_{\rm Haar} d\psi d\phi \tr_A\left(\hat{\omega} \ketbra{\psi^{\otimes m}, \phi^{\otimes m}}{\psi^{\otimes m}, \phi^{\otimes m}}\hat{\tau}_K^\dagger \right)\\
    &=\sum_{D_2 \times D_3} \calF_{\rm Haar}^{(1)} \calF_{\rm Haar}^{(K-1)}\\
    &=\frac{K^2}{d_A}\calF_{\rm Haar}^{(K-1)}
    = \frac{K(d_A+K-1)}{d_A} \calF_{\rm Haar}^{(K)}.
\end{align}
Combining them together, we obtain the summation for $\sigma_1 \in G_1$ in Eq.~\eqref{eq:perm_sum_G1} as
\begin{align}
    \sum_{\substack{\sigma_1 \in G_1 \\
    \sigma_2 \in G_K'}}\tr_A(\hat{\sigma}_1^\dagger)\tr_A(\hat{\sigma}_2^\dagger \hat{\sigma}_1) &=  \left(\frac{(m+d_A-1)!}{(d_A-1)!}\right)^{4} \left(\sum_{D_1\times D_3} + \sum_{D_2 \times D_4} + \sum_{D_1 \times D_4} + \sum_{D_2 \times D_3}\right) \tr_A\left(\rho_{\rm Haar}^{(m)}{}^{\otimes 2}\hat{\omega}^\dagger\right) \tr_A\left(\hat{\omega} \rho_{\rm Haar}^{(m)}{}^{\otimes 2} \hat{\tau}^\dagger_K\right) \nonumber\\
    &= \left(\frac{(m+d_A-1)!}{(d_A-1)!}\right)^{4} \left(\frac{K(d_A+K-1)}{d_A} + \frac{2 n K}{d_A} + \frac{n^2(K+1)}{d_A(d_A + K)} \right)\calF_{\rm Haar}^{(K)}\\
    &= \left(\frac{(m+d_A-1)!}{(d_A-1)!}\right)^{4} \left(K + \calO\left(\frac{K(K-1+2n)}{d_A}\right) \right)\calF_{\rm Haar}^{(K)}. 
\end{align}
Therefore, we see that only the case $\omega 
$ involves swaps to exchange replicas from $D_2$ and $D_3$ contributes in the leading order of $d_A$ because $\tr_A(\hat{\sigma}_1^\dagger)$ becomes additional frame potential and it can only be compensated by reducing the equal number of swaps in $\tau_K$. Following the same spirit, we can write out the summation over $\sigma_1 \in G_r$ for arbitrary $r\in [1, K-1]$ as
\begin{align}
   \sum_{\substack{\sigma_1 \in G_r \\
    \sigma_2 \in G_K'}} \tr_A(\hat{\sigma}_1^\dagger)\tr_A(\hat{\sigma}_2^\dagger \hat{\sigma}_1) &= \left(\frac{(m+d_A-1)!}{(d_A-1)!}\right)^{4}\left(\binom{K}{r}^2 \calF_{\rm Haar}^{(r)}\calF_{\rm Haar}^{(K-r)} + \calO\left(\frac{{\rm poly}(n, K)}{d_A}\right)\calF_{\rm Haar}^{(K)} \right)\\
    &= \left(\frac{(m+d_A-1)!}{(d_A-1)!}\right)^{4} \left(\binom{K}{r}\frac{(d_A-1)!(d_A+K-1)!}{(d_A+K-1-r)!(d_A-1+r)!} + \calO\left(\frac{{\rm poly}(n,K)}{d_A}\right)\right)\calF_{\rm Haar}^{(K)}.
    \label{eq:perm_sum_Gr_small}
\end{align}
Similarly, for $\sigma_1 \in G_K \setminus G_K'$, we have
\begin{align}
    \sum_{\substack{\sigma_1 \in G_K\setminus G_K' \\
    \sigma_2 \in G_K'}} \tr_A(\hat{\sigma}_1^\dagger)\tr_A(\hat{\sigma}_2^\dagger \hat{\sigma}_1) &= \left(\frac{(m+d_A-1)!}{(d_A-1)!}\right)^{4} \sum_\omega \tr_A\left(\rho_{\rm Haar}^{(m)}{}^{\otimes 2}\hat{\omega}^\dagger\right) \tr_A\left(\hat{\omega} \rho_{\rm Haar}^{(m)}{}^{\otimes 2} \hat{\tau}^\dagger_K\right)\\
    &= \left(\frac{(m+d_A-1)!}{(d_A-1)!}\right)^{4} \calF_{\rm Haar}^{(K)} \sum_\omega \tr_A\left(\hat{\omega} \rho_{\rm Haar}^{(m)}{}^{\otimes 2} \hat{\tau}^\dagger_K\right).
\end{align}
The leading order contribution of the above summation requires $\omega$ with $K-1$ disjoint swaps across $B_2$ and $B_3$ to compensate the effect of swaps from $\hat{\tau}_K^\dagger$ and an additional swap across $B_1\times B_3, B_2 \times B_4, B_1 \times B_4$. Furthermore, one can find that for either $B_1\times B_3$ or $B_2 \times B_4$, $\sum_\omega \tr_A\left(\hat{\omega} \rho_{\rm Haar}^{(m)}{}^{\otimes 2} \hat{\tau}^\dagger_K\right) = nK^2\calF_{\rm Haar}^{(1)}$ while for $B_1 \times B_4$, the summation becomes $n^2 K^2 \calF_{\rm Haar}^{(2)}$. Therefore we can conclude that
\be
    \sum_{\substack{\sigma_1 \in G_K\setminus G_K' \\
    \sigma_2 \in G_K'}} \tr_A(\hat{\sigma}_1^\dagger)\tr_A(\hat{\sigma}_2^\dagger \hat{\sigma}_1) = \left(\frac{(m+d_A-1)!}{(d_A-1)!}\right)^{4} \calO\left(\frac{{\rm poly}(n,K)}{d_A}\right) \calF_{\rm Haar}^{(K)},
\ee
which is a higher order term compared to the one Eq.~\eqref{eq:perm_sum_Gr_small} for arbitrary $r\in [1, m-1]$ thus can be omitted in following calculation.
Therefore, we can calculate Eq.~\eqref{eq:Fk_pseudo_HDT_lb_3} as
\begin{align}
    &d_B^{2t-1} \sum_{\substack{\sigma_1 \in \left(\cup_{r=0}^{K} G_r\right)\setminus G_K' \\
    \sigma_2, \dots, \sigma_t \in G_K'}} \prod_{i=1}^{t-1} (\hat{\sigma}_{i+1}|\hat{\sigma}_{i})_A (\hat{\sigma}_1|e)_A \nonumber\\
    &= d_B^{2t-1}  \left(\frac{(m+d_A-1)!}{(d_A-1)!}\right)^{2(t-2)} \left[ \left(\frac{(m+d_A-1)!}{(d_A-1)!}\right)^{4} \calF_{\rm Haar}^{(K)} +  \sum_{\substack{\sigma_1 \in \left(\cup_{r=1}^{K} G_r\right) \setminus G_K' \\
    \sigma_2 \in G_K'}} \tr_A(\hat{\sigma}_1^\dagger)\tr_A(\hat{\sigma}_2^\dagger \hat{\sigma}_1)\right] \nonumber\\
    &= d_B^{2t-1}  \left(\frac{(m+d_A-1)!}{(d_A-1)!}\right)^{2(t-2)} \left(\frac{(m+d_A-1)!}{(d_A-1)!}\right)^{4}\left[1 + \sum_{r=1}^{K-1} \binom{K}{r} \frac{(d_A-1)!(d_A+K-1)!}{(d_A+K-1-r)!(d_A-1+r)!} + \calO\left(\frac{{\rm poly}(n, K)}{d_A}\right)  \right]  \calF_{\rm Haar}^{(K)}\\
    &= d_B^{2t-1}  \left(\frac{(m+d_A-1)!}{(d_A-1)!}\right)^{2t} \left[1 + \frac{(d_A-1)! (2d_A+2K-2)!}{(d_A+K-1)! (2 d_A+K-2)!}-2 + \calO\left(\frac{{\rm poly}(n, K)}{d_A}\right)  \right]  \calF_{\rm Haar}^{(K)}\\
    &= d_B^{2t-1}  \left(\frac{(m+d_A-1)!}{(d_A-1)!}\right)^{2t} \left(\frac{(d_A-1)! (2d_A+2K-2)!}{(d_A+K-1)! (2 d_A+K-2)!}-1 + \calO\left(\frac{{\rm poly}(n, K)}{d_A}\right)  \right)  \calF_{\rm Haar}^{(K)},
    \label{eq:Fk_pseudo_HDT_lb_3sol}
\end{align}
where in the third line we apply the result in Eq.~\eqref{eq:perm_sum_Gr_small}. 

Finally, we can combine Eq.~\eqref{eq:Fk_pseudo_HDT_lb_1}, \eqref{eq:Fk_pseudo_HDT_lb_2} and \eqref{eq:Fk_pseudo_HDT_lb_3sol} to obtain the lower bound of pseudo frame potential (Eq.~\eqref{eq:Fk_pseudo_HDT_lb}) in 1-DW picture as
\begin{align}
    &\E_{\rm Haar}\calF_{\rm HDT}^{(n,K)}(t) \nonumber\\
    &\ge {\rm Wg}(e,2m;d)^t \left[d_B^t   \left(\frac{(m+d_A-1)!}{(d_A-1)!}\right)^{2t} + d_B^{2t} \left(\frac{(m+d_A-1)!}{(d_A-1)!}\right)^{2t} \calF_{\rm Haar}^{(K)} \nonumber \right.\\
    &\left. \qquad \qquad \qquad \qquad +d_B^{2t-1}  \left(\frac{(m+d_A-1)!}{(d_A-1)!}\right)^{2t} \left(\frac{(d_A-1)! (2d_A+2K-2)!}{(d_A+K-1)! (2 d_A+K-2)!}-1 + \calO\left(\frac{{\rm poly}(n, K)}{d_A}\right)  \right)  \calF_{\rm Haar}^{(K)}\right]\\
    &= \frac{1}{d_A^{2t}} \left(\frac{(m+d_A-1)!}{(d_A-1)!}\right)^{2t} \left[\frac{1}{d_B^t} + \calF_{\rm Haar}^{(K)} + \frac{1}{d_B}\left(\frac{(d_A-1)! (2d_A+2K-2)!}{(d_A+K-1)! (2 d_A+K-2)!}-1 + \calO\left(\frac{{\rm poly}(n,K)}{d_A}\right)\right)\calF_{\rm Haar}^{(K)}\right].
\end{align}
Take the replica trick $n \to 1-K, m \to K$, we obtain the lower bound for $K$-th order frame potential in HDT as
\be
    \E_{\rm Haar}\calF_{\rm HDT}^{(K)}(t) \ge \frac{1}{d_B^t} + \calF_{\rm Haar}^{(K)} + \frac{1}{d_B}\left(\frac{(d_A-1)! (2d_A+2K-2)!}{(d_A+K-1)! (2 d_A+K-2)!}-1 + \calO\left(\frac{{\rm poly}(K)}{d_A}\right)\right)\calF_{\rm Haar}^{(K)}.
    \label{eq:Fk_HDT_lb}
\ee
In particular, in the thermodynamic limit of $d_A \to \infty$, this lower bound can be reduced to
\be
   \E_{\rm Haar}\calF_{\rm HDT}^{(K)}(t) \ge \frac{1}{d_B^t} + \left(1 + \frac{2^K-1}{d_B} + \calO\left(\frac{{\rm exp}(K)}{d_A d_B}\right)\right)\calF_{\rm Haar}^{(K)},
\ee
which is Ineq.~\ref{high_order_F} presented in the main text.

The minimum number of ancilla in HDT to achieve $\epsilon$-approximate $K$-design ($\calF_{\rm DT}^{(K)} \le (1+\epsilon)\calF_{\rm Haar}^{(K)}$) is
\begin{align}
    N_B \ge \log_2\left(\frac{2^K-1}{\epsilon}\right) = K + \log_2(1/\epsilon) + \log_2\left(1-2^{-K}\right).
\end{align}

\section{Average mutual information in deep thermalization}
\label{app:mutual_info}

In this section, we derive the lower bound for the measurement-averaged mutual information in regular DT. We begin by connecting the average mutual information to average purity as follows.
\begin{align}
    \E_{\bfz}I(R:A_{\rm out}|{\bfz}) &\equiv \E_{\bfz} \left[S(\rho_{R|\bfz}) + S(\rho_{A_{\rm out}|\bfz}) - S(\rho_{RA_{\rm out}|\bfz})\right] \label{eq:avgMI_supp1}\\
    &= \E_{\bfz} 2S(\rho_{A_{\rm out}|\bfz})  \label{eq:avgMI_supp2}\\
    & \ge \E_{\bfz} 2S_2(\rho_{A_{\rm out}|\bfz}) = -2\E_{\bfz} \log_2\left(\tr(\rho_{A_{\rm out}|\bfz}^2)\right) \label{eq:avgMI_supp3}\\
    & \ge -2\log_2\left(\E_{\bfz} \tr(\rho_{A_{\rm out}|\bfz}^2)\right), \label{eq:avgMI_supp4}
\end{align}
where in Eq.~\eqref{eq:avgMI_supp1} $S(\cdot)$ is the von Neumann entanglement entropy and $\rho_{A_{\rm out}|\bfz} = \tr_{R}(\ketbra{\psi_{\bfz}}{\psi_{\bfz}}_{R A_{\rm out}})$ is the reduced state and similarly for the other one. Eq.~\eqref{eq:avgMI_supp2} follows the fact that the conditional output state $\rho_{RA_{\rm out}|\bfz} = \ketbra{\psi_{\bfz}}{\psi_{\bfz}}_{RA_{\rm out}}$ is a pure state. We utilize the monotonicity in R\'enyi entropy to obtain
Ineq.~\eqref{eq:avgMI_supp3}. The last inequality of Eq.~\eqref{eq:avgMI_supp4} originates from the concavity of the logarithmic function. Due to the same reason, we can further lower bound the Haar-averaged of Eq.~\eqref{eq:avgMI_supp4} as
\be
\E_\text{Haar}\E_{\bfz}I(R:A_{\rm out}|{\bfz}) \ge -2\E_\text{Haar}\log_2\left(\E_{\bfz} \tr(\rho_{A_{\rm out}|\bfz}^2)\right) \ge -2\log_2\left(\E_\text{Haar}\E_{\bfz} \tr(\rho_{A_{\rm out}|\bfz}^2)\right).
\ee

Now we evaluate the quantity of interest $\E_\text{Haar}\E_{\bfz} \tr(\rho_{A_{\rm out}|\bfz}^2)$. Through the global unitary $U$, the output pre-measurement state is
\be
\ket{\psi}_{RA_{\rm out} B} = (\bI_{R}\otimes U)\ket{\Phi}_{RA_\textup{in}}\ket{0}_B,
\ee
where $\ket{\Phi}_{RA_\text{in}} = \frac{1}{\sqrt{d_A}}\sum_{i=0}^{d_A-1} \ket{i}_R\ket{i}_{A_{\rm out}}$ is the maximally entangled state between $R$ and $A_\text{in}$. Here the unitary is only applied on systems of $A_\text{in}B$ and $R$ serves as a reference system for tracking the initial quantum information, as illustrated in Fig.~\ref{fig:MI}a1 of the main text. Via the projective measurement on the ancilla $B$, the conditional post-measurement state is
\be
    \ket{\psi_{\bfz}}_{RA_{\rm out}} = {}_B\braket{\bfz|\psi_\bfz}_{RA_{\rm out} B}/\sqrt{p_\bfz},
\ee
where $p_\bfz$ is the corresponding measurement probability
\be
    p_\bfz = \braket{\psi_\bfz|\psi_\bfz}_{RA_{\rm out}} = \tr\left(U\left(\ketbra{\Phi}{\Phi}_{RA_{\rm in}} \otimes\ketbra{0}{0}_B\right) U^\dagger \ketbra{\bfz}{\bfz}_B\right). 
\ee
The reduced state $\rho_{A_{\rm out}}$ is thus
\begin{align}
\rho_{A_{\rm out}|\bfz} &= \tr_{A_{\rm out}}(\ketbra{\psi_\bfz}{\psi_\bfz}_{RA_{\rm out}})\\
&= p_\bfz^{-1}\tr_{R B}(U \left(\ketbra{\Phi}{\Phi}_{RA_{\rm in}} \otimes\ketbra{0}{0}_B\right) U^\dagger \ketbra{\bfz}{\bfz}_B)\\
&= p_\bfz^{-1} \tr_{B}\left(U \left(\frac{\bI_{A_{\rm in}}}{d_A} \otimes \ketbra{0}{0}_B\right) U^\dagger \ketbra{\bfz}{\bfz}_B \right) \equiv p_\bfz^{-1} \tilde{\rho}_{A_{\rm out}|\bfz},
\end{align}
where we define $\tilde{\rho}_{A_{\rm out}|\bfz}$ to be the unnormalized state for convenience, and thus the measurement probability is simply $p_\bfz = \tr(\tilde{\rho}_{A_{\rm out}|\bfz})$. With these definitions, we now compute the average purity given by
\be
\gamma \equiv \E_{\bfz}\tr(\rho_{A_{\rm out}|\bfz}^2) =  \sum_\bfz p_\bfz^{-1} \tr(\tilde{\rho}_{A_{\rm out}|\bfz}^2)
= \sum_\bfz \left[\tr(\tilde{\rho}_{A_{\rm out}|\bfz})^{-1}\tr(\tilde{\rho}_{A_{\rm out}|\bfz}^{\otimes 2} \hat{\tau}_1)\right],
\ee
where $\hat{\tau}_1$ is the swap operator between the two replicas as we have seen in Appendix \ref{app:F_dynamics}. We utilize the replica trick and first consider the pseudo purity $\gamma^{(n)}$ as
\begin{align}
    \gamma^{(n)} &\equiv \sum_\bfz \tr(\tilde{\rho}_{A_{\rm out}|\bfz})^n \tr(\tilde{\rho}_{A_{\rm out}|\bfz}^{\otimes 2} \hat{\tau}_1)\\
    &= \sum_\bfz \tr(\tilde{\rho}_{A_{\rm out}|\bfz}^{\otimes (n+2)} \bI_n \otimes \hat{\tau}_1)\\
    &= \sum_\bfz \tr(\tr_{B}\left(U \left(\frac{\bI}{d_A} \otimes\ketbra{0}{0}_B\right) U^\dagger \ketbra{\bfz}{\bfz}_B \right)^{\otimes (n+2)} \bI_n \otimes \hat{\tau}_1 )\\
    &= \sum_\bfz \tr(U^{\otimes (n+2)} \left(\frac{\bI^{\otimes (n+2)}}{d_A^{(n+2)}} \otimes \ketbra{0}{0}_B^{\otimes (n+2)}\right) U^\dagger{}^{\otimes (m+2)} \left(\bI_n \otimes \hat{\tau}_1\right)_{A_{\rm out}} \otimes \ketbra{\bfz}{\bfz}_B^{\otimes (n+2)})\\
    &= \frac{1}{d_A^{n+2}} {}_{A}(\hat{e}| {}_B(0|\left[U^{\otimes (n+2)} \otimes U^*{}^{\otimes (n+2)}\right] |\hat{e}_n\hat{\tau}_1)_{A} |\hat{H}_{n+2})_B,
    \label{eq:avgMI_supp5}
\end{align}
where we define $\hat{H}_{n+2} \equiv \sum_\bfz \ketbra{\bfz}{\bfz}^{\otimes (n+2)}$, and utilize the doubled Hilbert space representation in the last line. Here we utilize $\hat{e}_n$ to represent identity operator on the $n$ replicas.

As we have seen from Eq.~\eqref{eq:avgMI_supp5}, the boundary condition is $\hat{e}_n\otimes \hat{\tau}_1$ and $\hat{H}_{n+2}$ on subsystem $A$ and $B$ separately. 
The ensemble average of $\gamma^{(n)}$ can be written as utilizing Eq.~\eqref{eq:haar_integral}
\begin{align}
    \E_\text{Haar}\left[\gamma^{(n)}\right] &= \frac{1}{d_A^{n+2}} {}_A(\hat{e}| {}_B(0| \E_\text{Haar} \left[U^{\otimes (n+2)} \otimes U^*{}^{\otimes (n+2)}\right]|\hat{e}_n\hat{\tau}_1)_A|\hat{H}_{n+2})_B\\
    &= \frac{1}{d_A^{n+2}} {}_A(\hat{e}| {}_B(0| \left({\rm Wg}(e, n+2; d) \sum_{\sigma} |\hat{\sigma})(\hat{\sigma}| + \sum_{\sigma \neq \pi} {\rm Wg}(\sigma^{-1}\pi, n+2) |\hat{\sigma})(\hat{\pi}|\right) |\hat{e}_n\hat{\tau}_1)_A|\hat{H}_{n+2})_B. \label{eq:avgMI_supp7}
\end{align}
For the first term in Eq.~\eqref{eq:avgMI_supp7}, we have
\begin{align}
    &{}_A(\hat{e}| {}_B(0| {\rm Wg}(e, n+2; d) \sum_{\sigma \in S_{n+2}}  |\hat{\sigma})_{AB} (\hat{\sigma}|\hat{e}_n\hat{\tau}_1)_A (\hat{\sigma}|\hat{H}_{n+2})_B = {\rm Wg}(e, n+2;d) \sum_{\sigma \in S_{n+2}} \tr_A(\hat{\sigma}) \tr_A(\hat{\sigma}^\dagger \left(\bI_n\otimes \hat{\tau}_1\right)) d_B \label{eq:avgMI_supp8}\\
    &= {\rm Wg}(e, n+2;d) d_B \sum_{\sigma \in S_{n+2}} d_A^{|\sigma| + |\sigma^{-1}(e_n\otimes \tau_1)|}\\
    &= {\rm Wg}(e, n+2;d) d_B \left[d_A^{|e|+|e_n\otimes \tau_1|} + d_A^{|e_n\otimes \tau_1| + |e|} + \dots \right]\\
    &= {\rm Wg}(e, n+2;d) d_B \left(2 d_A^{2n+3} + \calO\left(d_A^{2n+1}\right)\right) = \frac{d_A^{n+1}}{d_B^{n+1}}\left(2 + \calO\left(\frac{1}{d_A^2}\right)\right),
\end{align}
where in the first equation we utilize the fact that $(\hat{\sigma}|\hat{H}_{n+2}) = d_B, \forall \sigma \in S_{n+2}$, and $|\sigma|$ is again the number of cycles of permutation $\sigma$.  In the last equation we utilize ${\rm Wg}(e,n+2;d) = 1/d^{n+2}$ in the thermodynamic limit of $d_A \to \infty$ thus $d \to \infty$.
Similarly, the second term in Eq.~\eqref{eq:avgMI_supp7} becomes 
\begin{align}
    &{}_A(\hat{e}| {}_B(0| \sum_{\sigma \neq \pi} {\rm Wg}(\sigma^{-1}\pi, n+2;d)   |\hat{\sigma})_{AB} (\hat{\pi}|\hat{e}_n\hat{\tau}_1)_A (\hat{\pi}|\hat{H}_{n+2})_B = \sum_{\sigma \neq \pi} {\rm Wg}(\sigma^{-1}\pi, n+2;d) \tr_A(\hat{\sigma}) \tr_A(\hat{\pi}^\dagger (\hat{e}_n\otimes \hat{\tau}_1)) d_B\\
    &= d_B \sum_{\sigma \neq \pi} {\rm Wg}(\sigma^{-1}\pi, n+2;d) d_A^{|\sigma| + |\pi^{-1}(e_n\otimes \tau_1)|}.
    \label{eq:avgMI_supp6}
\end{align}
The above summation can be evaluated as follows.
\begin{align}
     &\sum_{\sigma \neq \pi} {\rm Wg}(\sigma^{-1}\pi, n+2;d) d_A^{|\sigma| + |\pi^{-1}(e_n\otimes \tau_1)|} \nonumber\\
     &= \sum_{\pi \neq e} {\rm Wg}(e\pi, n+2;d) d_A^{|e| + |\pi^{-1}(e_n\otimes \tau_1)|} + \sum_{\pi \neq e_n\otimes \tau_1}  {\rm Wg}((e_n\otimes \tau_1^{-1})\pi, n+2;d) d_A^{|e_n \otimes \tau_1| + |\pi^{-1}(e_n\otimes \tau_1)|} + \dots \\
     &\ge -\frac{(n+2)!}{d_A^{n+3} d_B^{n+3}} d_A^{2(n+2)} - \frac{(n+2)!}{d_A^{n+3} d_B^{n+3}}d_A^{2(n+1)}      + \calO\left(\frac{d_A^{2n+1}}{d_A^{n+4}d_B^{n+4}}\right)\\
     &= -\frac{d_A^{n-1}}{d_B^{n+3}}\left((n+2)!(d_A^2 + 1) + \calO\left(\frac{1}{d_A^2 d_B}\right)\right).
\end{align}
where in the inequality we utilize ${\rm Wg}(e(e_n\otimes \tau_1), n+2;d) = -1/d_A^{n+3}d_B^{n+3}$ and $\sum_{i=1}^N {\rm sgn}(x_i)|x_i| |y_i| \ge -N \max_i |x_i||y_i|$. The second term of Eq.~\eqref{eq:avgMI_supp6} thus becomes
\begin{align}
    {}_A(\hat{e}| {}_B(0| \sum_{\sigma \neq \pi} {\rm Wg}(\sigma^{-1}\pi, n+2;d)   |\hat{\sigma})_{AB} (\hat{\pi}|\hat{e}_n\hat{\tau}_1)_A (\hat{\pi}|\hat{H}_{n+2})_B =-\frac{d_A^{n-1}}{d_B^{n+2}}\left((n+2)!(d_A^2 + 1) + \calO\left(\frac{1}{d_A^2 d_B}\right)\right) .
\end{align}
Therefore, the ensemble average of pseudo purity in Eq.~\eqref{eq:avgMI_supp7} becomes
\begin{align}
    &\E_\text{Haar}\left[\gamma^{(n)}\right] = \frac{1}{d_A^{n+2}}\left[\frac{d_A^{n+1}}{d_B^{n+1}}\left(2 + \calO\left(\frac{1}{d_A^2}\right)\right) - \frac{d_A^{n-1}} {d_B^{n+2}}\left((n+2)!(d_A^2+1) + \calO\left(\frac{1}{d_A^2 d_B}\right)\right)\right]\\
    &=\frac{2}{d_A d_B^{n+1}} - \frac{(n+2)!}{d_Ad_B^{n+2}}  +\calO\left(\frac{1}{d_A^3 d_B^{n+1}}\right),
\end{align}
and by taking the replica trick $n \to -1$, the ensemble average purity is
\be
    \E_\text{Haar}\left[\gamma\right] = \frac{2}{d_A} - \frac{1}{d_A d_B}  +\calO\left(\frac{1}{d_A^3}\right).
\ee

The expected average mutual information is thus lower bounded in the thermodynamic limit $d_A \to \infty$ by 
\be
\E_\text{Haar}\E_\bfz I(R:A_{\rm out}|z) \ge -2\log_2\left(\frac{2}{d_A} - \frac{1}{d_A d_B}  +\calO\left(\frac{1}{d_A^3}\right)\right),
\ee
which is the Lemma~\ref{lemma:security_DT} in the main text.

\section{Additional details on quantum circuit size}
\label{app:qVol}

In this section, we discuss the scaling of quantum circuit size with different choices of $N_B$ in HDT. For convenience, we reprint the definition of quantum circuit size here
\begin{align}
    \mbox{Q-size} &\equiv  \mbox{ depth per step}\times T \times (N_A+N_B)\\
    &\propto T \times (N_A + N_B)^2,
\end{align}
where the second line originates from the unitary design~\cite{brandao2016local, brandao2017quantum}.

To generate an $\epsilon$-approximate $K$-design state ensemble in HDT, we need $N_B \ge K + \log_2(1/\epsilon)$. From Ineq.~\eqref{NB_T_tradeoff} in the main text, the minimum Q-size is
\begin{align}
&\mbox{Q-size} \propto \frac{(N_A+N_B)^2}{N_B}
\times\left[K N_A - \log_2(K!) + \log_2 \left(\frac{1}{\epsilon}\right)\right].
\label{eq:qvol}
\end{align}
Following Eq.~\eqref{eq:qvol}, we notice that the minimum quantum circuit size is achieved when $N_B = N_A$ (assuming that $N_A \ge K + \log_2(1/\epsilon)$),
\begin{align}
    \mbox{Q-size}\rvert_{N_B = N_A} &\gtrsim  4N_A \left[KN_A -\log_2(K!) + \log_2\left(\frac{1}{\epsilon-(2^K-1)2^{-N_A}}\right)\right]\\
    &\sim 4K N_A^2 + 4N_A \log_2 (1/K!\epsilon).
\end{align}

At the minimum necessary bath $N_B = K + \log_2(1/\epsilon)$, from the expression of the relative frame potential error,
\begin{align}
    &\delta^{(K)} = \frac{2^K-1}{d_B} + d_B^{-T} \frac{1}{\calF_\text{Haar}^{(K)}} = \frac{2^K-1}{d_B} + d_B^{-T} \frac{d_A^K}{K!} = \epsilon,
\end{align}
where the second equation holds at thermodynamic limit $d_A \to \infty$,
we can solve the necessary time steps
\begin{align}
T\rvert_{N_B = K+\log_2(1/\epsilon)} = K \frac{N_A}{N_B} - \frac{1}{N_B} \log_2(K!) + 1.
\end{align}
Then the quantum circuit size becomes
\begin{align}
    \mbox{Q-size}\rvert_{N_B = K+\log_2(1/\epsilon)} &\propto T (N_A+N_B)^2\\
    &= \left(K N_A - \log_2(K!) + K + \log_2(1/\epsilon) \right) \left(\frac{N_A}{K+\log_2(1/\epsilon)} + 1\right)(N_A + K+\log_2(1/\epsilon)).
    \label{eq:qvol_minB}
\end{align}
On the other hand, for $T=1$ corresponding to the conventional DT, to ensure a small error
\begin{align}
    &\delta^{(K)} = \frac{1}{d_B}\left(1+\frac{1}{\calF_\text{Haar}^{(K)}}\right) = \frac{1}{d_B}\left(1+\frac{d_A^K}{K!}\right) \le \epsilon,
\end{align}
we require
\begin{align}
    N_B \ge \log_2\left(1+\frac{d_A^K}{K!}\right) + \log_2(1/\epsilon) \simeq K N_A - \log_2(K!) + \log_2(1/\epsilon),
\end{align}
where the first line we utilize the result from Ref.~\cite{ippoliti2022solvable}, which is slightly more precise than our result and the second approximation holds when $d_A^K/K! \gg 1$. The quantum circuit size for conventional DT is
\begin{align}
    \mbox{Q-size}\rvert_{N_B = K N_A - \log_2(K!) + \log_2(1/\epsilon)} = \left[(K+1) N_A - \log_2(K!) + \log_2(1/\epsilon)\right]^2,
    \label{eq:qvol_dt}
\end{align}
which is always a quadratic scaling in $N_A$.

\begin{figure}[t]
    \centering
    \includegraphics[width=0.35\textwidth]{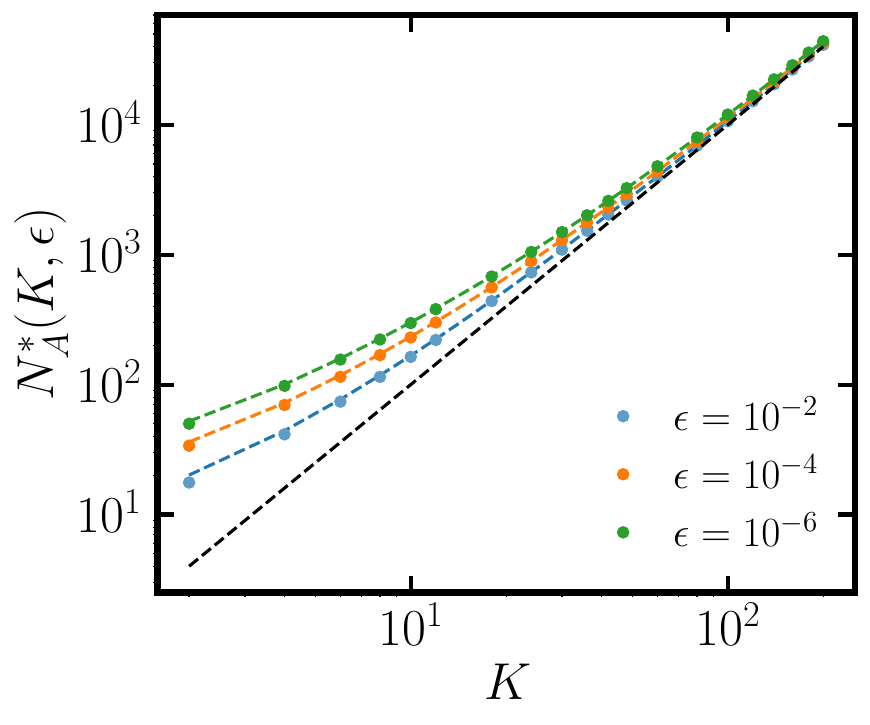}
    \caption{The critical system size $N_A^*(K, \epsilon)$ for quantum circuit size with minimum bath equals conventional deep thermalization. Colored dots represent $N_A^*(K, \epsilon)$ with different $\epsilon$. Colored dashed lines represent the solution from Eq.~\eqref{eq:N_crit} with corresponding $\epsilon$, and black dashed line indicate the scaling of $\sim K^2$.}
    \label{fig:turningpoint}
\end{figure}
To compare the quantum circuit size requirement in Eq.~\eqref{eq:qvol_minB} and Eq.~\eqref{eq:qvol_dt}, we can analytically solve the critical $N_A^*(K, \epsilon)$ such that their quantum circuit size is equal. To reveal the scaling of $N_A^*(K, \epsilon)$ explicitly, we have
\begin{align}
    &\left(K N_A - \log_2(K!) + K + \log_2(1/\epsilon) \right) \frac{(N_A + K+\log_2(1/\epsilon))^2}{K + \log_2(1/\epsilon)} = ((K+1) N_A - \log_2(K!) + \log_2(1/\epsilon))^2\\
    &\left(K N_A - K + K + \log_2(1/\epsilon) \right) \frac{(N_A + K+\log_2(1/\epsilon))^2}{K + \log_2(1/\epsilon)} = ((K+1) N_A + \log_2(1/\epsilon))^2\\
    &N_A^*(K, \epsilon) \simeq \frac{1}{2} \left(\sqrt{(K+\log_2(1/\epsilon)) \left(K^3 + K^2\log_2(1/\epsilon)+4\log_2(1/\epsilon)\right)}+K^2+K \log_2(1/\epsilon)\right)
    \label{eq:N_crit}
\end{align}
where in the second line we take $\log_2(K!) \to K$ in the L.H.S and get rid of the one in R.H.S. for convenience. Eq.~\eqref{eq:N_crit}
indicates that the critical system size $N_A^*(K, \epsilon)$ undergoes a continuous transition from approximation-error-dependent scaling $\sim K \log_2(1/\epsilon)$ to universal scaling $K^2$ with increasing order $K$ shown in Fig.~\ref{fig:turningpoint}, and we can roughly estimate the transition point to be $K \sim \log_2(1/\epsilon)$. 
The threshold can be summarized as
\be 
 N_A^*(K, \epsilon)=\left\{ \begin{array}{ll}
        K \log_2(1/\epsilon), & \mbox{if $K \ll \log_2(1/\epsilon)$};\\
        K^2, & \mbox{if $K \gg \log_2(1/\epsilon)$}.\end{array} \right.
\ee 
When $N_A\ge  N_A^*(K, \epsilon)$, the minimum bath case requires higher quantum circuit size, otherwise conventional deep thermalization requires higher quantum circuit size. 

In the following, we analyze the scaling of quantum circuit size with $K \lessgtr \log_2 (1/\epsilon)$. 

For $K < \log_2 (1/\epsilon)$, the quantum circuit size of minimum bath scales as
\begin{enumerate}
    \item $N_A < N_A^* = K \log_2 (1/\epsilon)$. $\mbox{Q-size} \simeq \left(K + 1\right)\left(K N_A  + \log_2(1/\epsilon) \right) (N_A+\log_2(1/\epsilon))$.
    
    \item $N_A > N_A^* = K \log_2 (1/\epsilon)$. $ \mbox{Q-size} \simeq \frac{K N_A^3}{K+\log_2(1/\epsilon)}$.
\end{enumerate}
On the other hand for $K > \log_2(1/\epsilon)$, the quantum circuit size of minimum bath scales as
\begin{enumerate}
    \item $N_A < N_A^* = K^2$. $\mbox{Q-size} \simeq \left(K N_A - \log_2(K!) + K\right) \frac{K^2}{K + \log_2(1/\epsilon)}(N_A + K) \simeq \frac{K^3}{K + \log_2(1/\epsilon)}(N_A + 1 - \log_2 K)(N_A + K)$.
    
    \item $N_A > N_A^* = K^2$. $ \mbox{Q-size} \simeq \frac{K N_A^3}{K+\log_2(1/\epsilon)}$.
\end{enumerate}
where we apply $\log_2(K!) \simeq K \log_2(K)$ in the first line. We demonstrate the above scalings in Fig.~\ref{fig:qVol_minbath}a and b for the two cases of $K \lessgtr \log_2(1/\epsilon)$ separately.

\begin{figure}[t]
    \centering
    \includegraphics[width=0.5\textwidth]{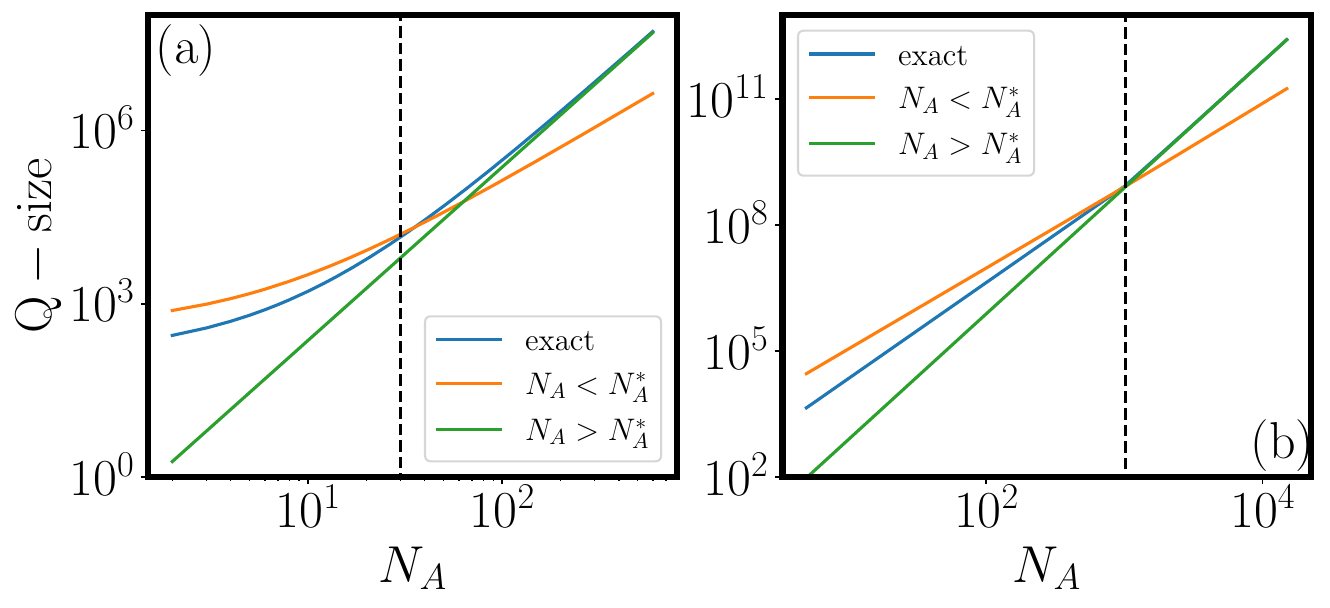}
    \caption{The scaling of quantum circuit size with minimum bath for (a) $K < \log_2(1/\epsilon)$ and (b) $K > \log_2 (1/\epsilon)$. Parameters: $\epsilon=10^{-3}$; (a) $K = 3$ (b) $K = 32$. Blue lines show exact result in Eq.~\eqref{eq:qvol_minB}. Orange lines represent scaling analyses as (a) $(K+1)(K N_A + \log_2(1/\epsilon))(N_A + \log_2(1/\epsilon))$ and (b) $ K^3 (N_A +1 - \log_2(K)) (N_A + K)/(K + \log_2(1/\epsilon))$. The green lines represent $K N_A^3/(K + \log_2 (1/\epsilon)$. Black vertical dashed lines indicate (a) $N_A^* \sim K \log_2(1/\epsilon)$ and (b) $N_A^* \sim K^2$. }
    \label{fig:qVol_minbath}
\end{figure}

In the end, we compare the scaling analysis of quantum circuit size with numerics for three cases of interest illustrated in Fig.~\ref{fig:qVol}c in the main text: minimum bath HDT, minimum Q-size HDT and conventional DT. In Fig.~\ref{fig:qVol_addition}a with large system $N_A \gg K\log_2(1/\epsilon)$, Q-size of the minimum Q-size HDT (orange dots) and DT (green dots) grows quadratically with the system size $N_A$, though reduces circuit size requirement by a $K$-dependent constant. However, for minimum bath HDT, its Q-size grows cubically with $N_A$. The numerical results agree with our theory analyses (dashed lines) summarized in Fig.~\ref{fig:qVol}c in the main text.
On the other hand, for relative small system $N_A \ll K \log_2(1/\epsilon)$, the minimum Q-size (orange) and conventional DT (green) remain quadratic growth with $N_A$, shown Fig.~\ref{fig:qVol_addition}b. However, the minimum bath case (blue) undergoes a continuous transition from quadratic toward cubic growth and requires smaller necessary resources compared to conventional DT.

\begin{figure}[t]
    \centering
    \includegraphics[width=0.5\textwidth]{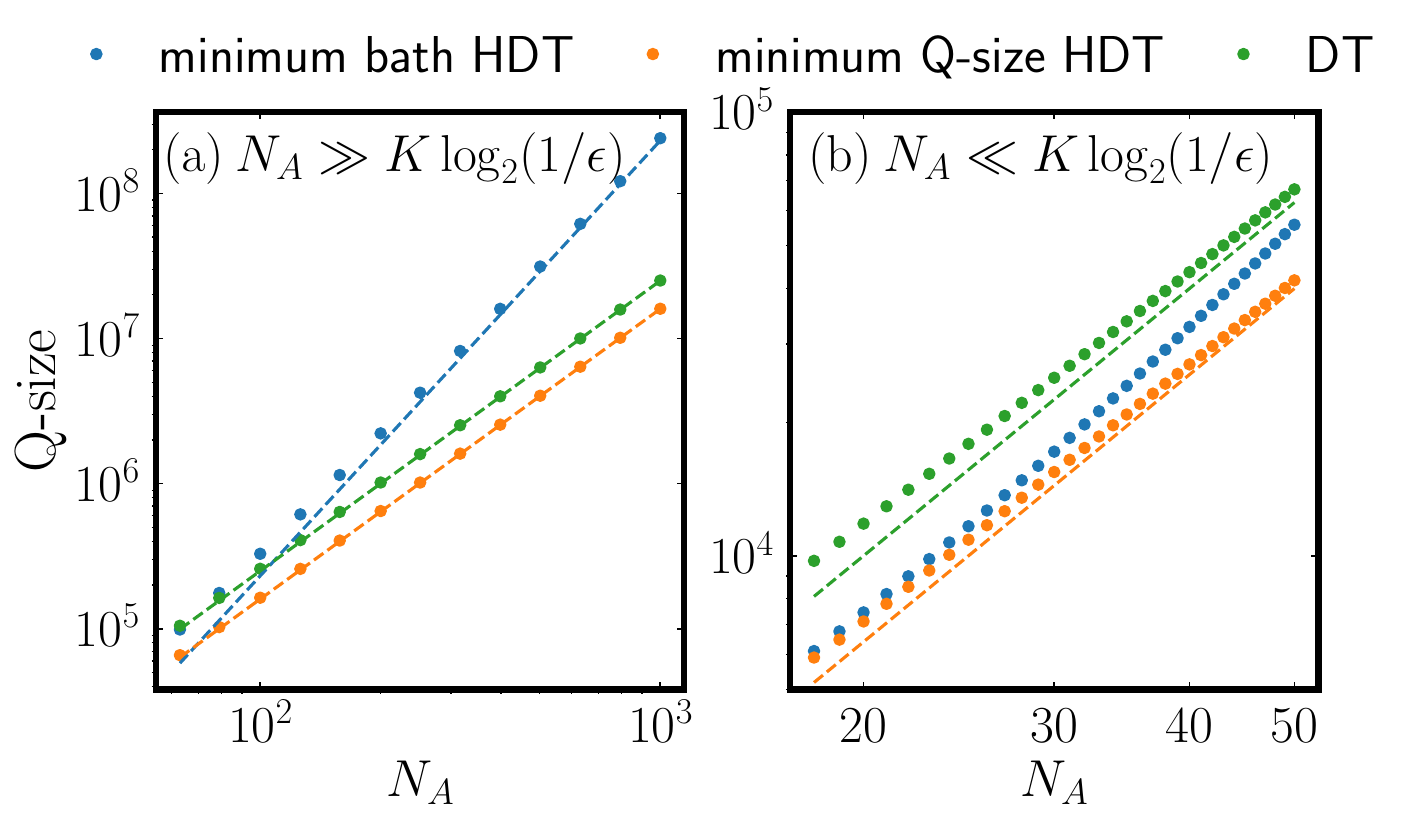}
    \caption{Scaling of Q-size for (a) large and (b) small system. In (a) we choose $K=4$ and $\epsilon=10^{-4}$ with bath size specified namely $N_B=17,N_A ,4N_A+13$ for minimum bath HDT, minimum Q-size HDT and conventional DT. Dots are numerically solved. Orange and green dashed lines in (a)-(b) represent scaling of $4KN_A^2$ and $(K+1)^2 N_A^2$. Blue dashed line in (a) represent $KN_A^3/(K + \log_2(1/\epsilon))$.}
    \label{fig:qVol_addition}
\end{figure}



\section{Additional numerical details}
\label{app:numerical_details}

\begin{figure}[t]
    \centering
    \includegraphics[width=0.75\textwidth]{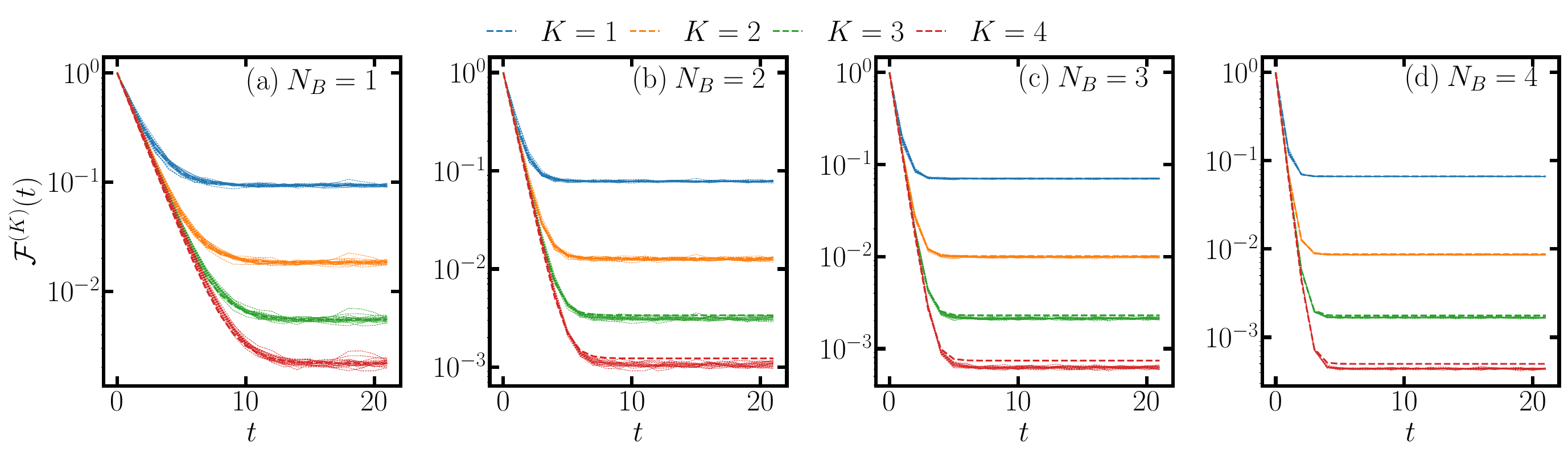}
    \caption{Dynamics of frame potential $\calF^{(K)}$ with different $K$ in a system of $N_A=4$ and $N_B=1,2,3, 4$ (from left to right). Thin dashed lines represent the dynamics of frame potential with each sequence of randomly sampled Haar unitaries, and the thick dashed line show corresponding theoretical prediction in Eq.~\eqref{F1_exact} ($K=1$) and Ineq.~\eqref{high_order_F} ($K>1$) in the main text.
    \label{fig:Fevo2}}
    \centering
    \includegraphics[width=0.45\textwidth]{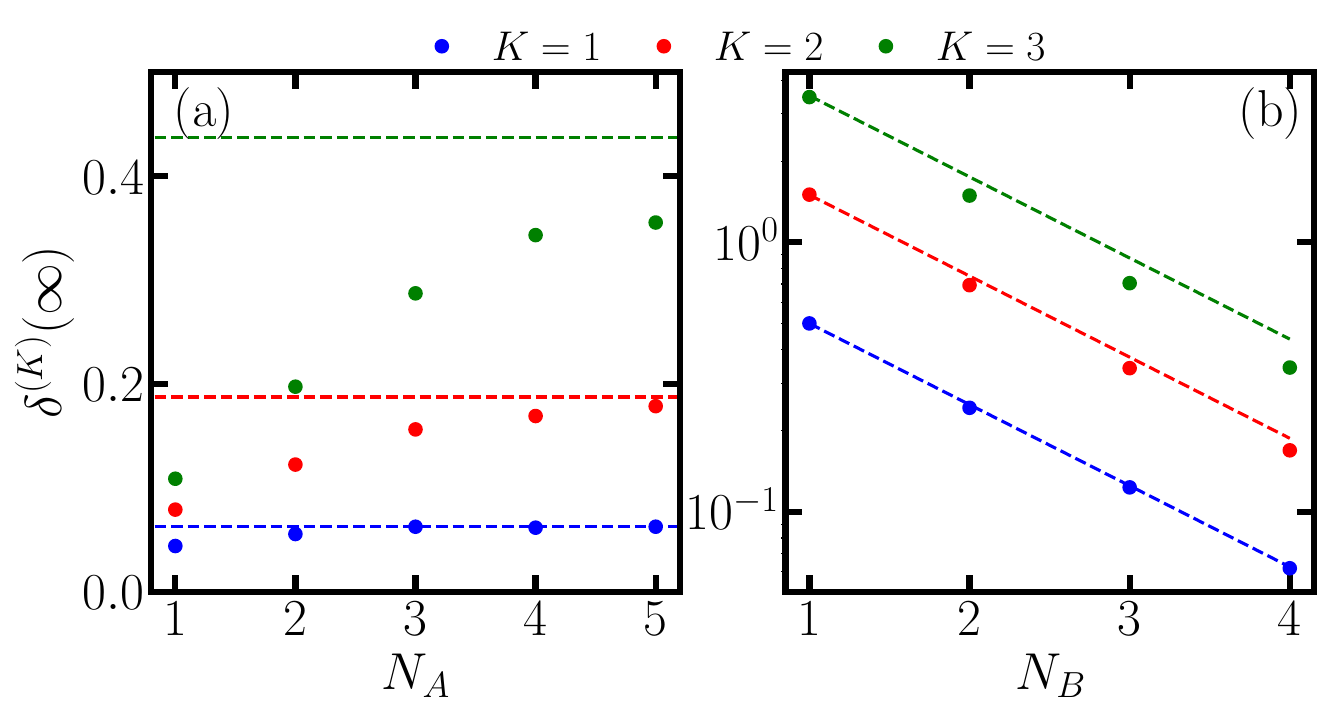}
    \caption{Relative deviation of converged frame potential $\delta^{(K)}(\infty) = \calF^{(K)}(\infty)/\calF^{(K)}_\text{Haar} - 1$. The system consists of (a) $N_B=4$ bath qubits with various $N_A$ data qubits and (b) $N_A=4$ data qubits with various $N_B$ bath qubits.
    Dots represents numerical simulation results and dashed lines represent theoretical prediction in Eq.~\eqref{F1_exact} (blue) and Ineq.~\eqref{high_order_F} in the main text (red and green) as $(2^K-1)/d_B$. Each dot is an average over $50$ Haar unitaries realizations with ensemble of $5\times 10^4$ samples at step $T=32$. Here we take the empirical $\calF^{(K)}_\text{Haar}$ by averaging over $50$ ensembles of $5\times 10^4$ Haar random states to reduce the finite size effect.
    \label{fig:F_inf}}
\end{figure}

In this section, we provide additional details on the numerical simulation results. In Fig.~\ref{fig:Fevo2}, we present the dynamics of the frame potential of \emph{projected ensemble} with each Haar unitary implementation in detail. 
Considering the different random unitary initialization (thin dashed lines), the results generally align with the theoretical prediction for ensemble average though with small fluctuations, indicating the Haar ensemble theory can well capture the dynamics for a typical realization of random unitaries in HDT.

We further show the relative deviation of converged frame potential $\delta^{(K)}(\infty) = \calF^{(K)}(\infty)/\calF^{(K)}_\text{Haar} - 1$ versus $N_A$ and $N_B$, which is a direct measure on the approximation error to $K$-design. In Fig.~\ref{fig:F_inf}a, given a bath of $N_B$ qubits, the relative deviation (dots) saturates towards a constant (dashed lines) independent of data system size $N_A$, though it increases with $K$ due to the greater randomness required for higher-order designs. Meanwhile, increasing the bath system size can reduce the deviation exponentially in Fig.~\ref{fig:F_inf}b as we have seen in Fig.~\ref{fig:Fevo} of the main text. Overall, the numerical results support our empirical result in Ineq.~\eqref{high_order_F} of the main text.



\section{Additional details on quantum machine learning}
\label{app:QML_details}
In this section, we provide additional details on the loss function and circuit design in the QML-enhanced HDT.

We have chosen an interpolated loss function
\begin{equation}\label{eq:loss_def}
    \begin{aligned}
    \calL_t(\rho_t)&=(1-\tau_t) \bigl(1-F_\text{sup} \left(\rho_t, \ket{\psi_0}\right) \bigr) + \tau_t \bigl(1-F_\text{sup} \left(\rho_t, \bI/d_A\right)\bigr),
    \end{aligned}
\end{equation}
where $\rho_t$ is the average state at step $t$ and $\bI/d_A$ is the maximally mixed state on the $N_A$ data qubit. We adopt the \emph{superfidelity}~\cite{cerezo2021sub} for efficient evaluation in training without significant impact on performance
\begin{align}
    F_{\rm sup}(\rho, \sigma) = \tr(\rho \sigma) + \sqrt{(1- \tr\rho^2) (1-\tr\sigma^2)}. 
\end{align}
Specifically, for the two states $\ket{\psi_0}$ and $\bI/d_A$ considered in loss function, the \emph{superfidelity} can be further simplified to
$F_\text{sup}(\rho_t, \ket{\psi_0}) = \braket{\psi_0|\rho_t|\psi_0}$ and $
    F_\text{sup}(\rho_t, \bI/2^N) = 1/d_A + \sqrt{(d_A - 1)(1- \tr\rho_t^2)/d_A}$.
Therefore, one can regard the loss function in Eq.~\eqref{eq:loss_def} as a nonlinear schedule for the decay of purity of $\rho_t$ towards $1/d_A$. The loss function considered here is also closely connected to the MMD distance in QuDDPM~\cite{zhang2024generative}.

Note that utilizing the average-state-dependent loss function avoids the precise control and tracking of measurement results on ancillary qubits, though a Monte-Carlo sampling error exists due to finite sampling. In particular, both $\braket{\psi_0|\rho_t|\psi_0}$ and $\tr(\rho_t^2)$ are experimental tractable to measure. For the purity measure, one can either run two quantum devices simultaneously to obtain two copies of the average state, and then estimate the purity via a control-SWAP circuit with an extra control qubit. Another recent proposed option relies on the random measurement toolbox, e.g., classical shadow~\cite{huang2020predicting}. For the control-SWAP circuit, it requires the application of a global control-SWAP gate on the two copies of $N$-qubit state conditioned on control qubit's state, which can be decomposed into $N$ control-SWAP gates applied on two qubits only. In comparison, classical shadow only requires an extra layer of random single qubit Pauli gate for efficient estimation, though at the cost of more random measurements. 

In each step, we parameterize the unitary via a hardware-efficient ansatz~\cite{kandala2017hardware} where each layer of circuit involves single qubit X and Y rotations followed by control-Z gates in brickwall style. In numerical simulation, we choose $L = 2(N_A + N_B)$ layers for each step.

\section{Additional details on experiment}

In this section, we provide additional details on the IBM Quantum experiments. In Fig.~\ref{fig:tradeoff_detail}a-c, we show the results of the first-order frame potential for each choice of randomly selected circuit parameters (grey dashed lines). Despite fluctuations due to finite system size and randomly selected parameters, the decay of $\calF^{(1)}(T)$ versus $N_B$ and $T$ is well captured by the Haar ensemble average results.

\begin{figure*}[t]
    \centering
    \includegraphics[width=0.65\textwidth]{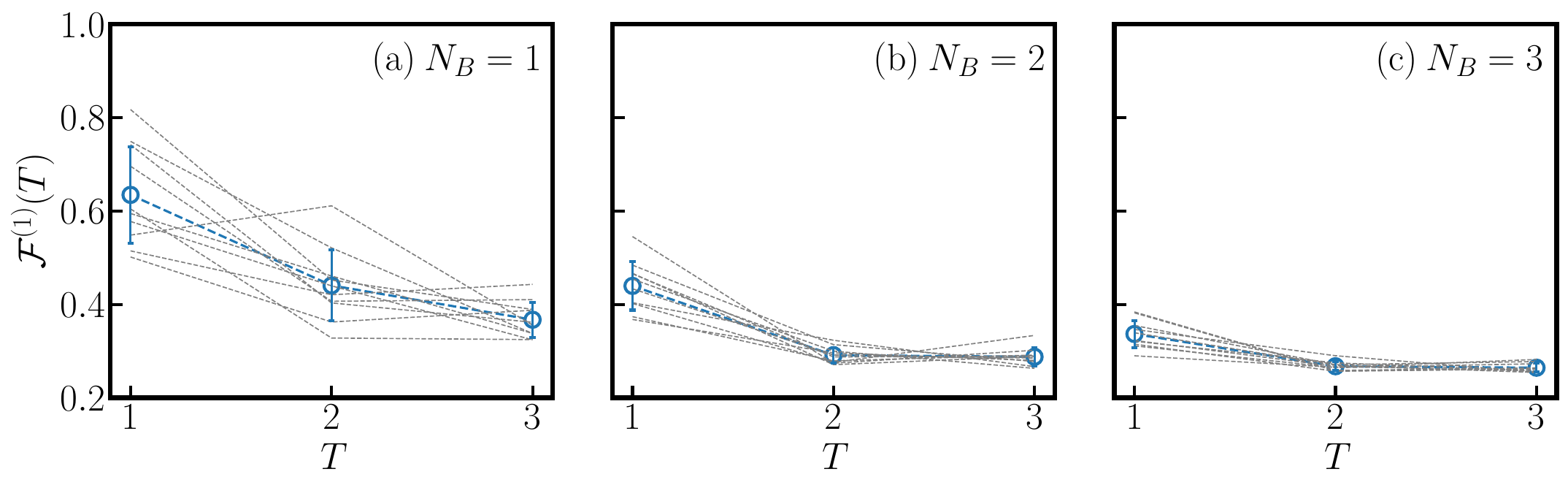}
    \caption{Decay of first order frame potential in a system of $N_A=2$ data qubits and various $N_B$ bath qubits versus time steps $T$ of IBM Quantum Brisbane. Grey dashed line represent result with each set of parameters. Blue circles with errorbars show the mean and standard deviation. 
    \label{fig:tradeoff_detail}}
\end{figure*}

\end{widetext}

\end{document}